%% file: main.tex
\documentclass{article}

\PassOptionsToPackage{round}{natbib}

\usepackage{colm2024_conference}
\colmfinalcopy

\usepackage{hyperref}
\usepackage{algorithmic}
\usepackage{longtable}
\usepackage{wrapfig}

\newcounter{excounter}[section]
\renewcommand{\theexcounter}{\thesection.\arabic{excounter}}

\usepackage[utf8]{inputenc} 
\usepackage[T1]{fontenc}    
\usepackage{hyperref}       
\usepackage{url}            
\usepackage{booktabs}       
\usepackage{amsfonts}       
\usepackage{nicefrac}       
\usepackage{microtype}      
\usepackage{subcaption} 
\usepackage[symbol]{footmisc}
\usepackage{tcolorbox}

\hypersetup{
    colorlinks,
    linkcolor={red!50!black},
    citecolor={blue!60!black},
    urlcolor={blue!80!black}
}

\colorlet{darkred}{red!60!black}

\usepackage{algorithm}
\usepackage{enumitem}
  \setlist{leftmargin=*}
\usepackage{graphicx}
  \graphicspath{ {./figs/} }
\usepackage{lineno}
  \linenumbers
\usepackage{YK}
\usepackage{wrapfig}  

\title{Estimation with missing not at random binary outcomes via exponential tilts}

\author{%
Subha Maity \\
Department of Statistics and Actuarial Science \\
University of Waterloo \\
\texttt{smaity@uwaterloo.ca} \\
}

\begin{document}
\nolinenumbers

\maketitle




  %
\begin{abstract}
  \input{sections/abstract}
\end{abstract}


\input{sections/intro}
\input{sections/setup}

\input{sections/estimation}

\input{sections/experiments}

\input{sections/discussion}


\bibliography{sm,extra}
\bibliographystyle{abbrvnat}

\newpage
\appendix
\input{sections/emprical-likelihood}
\input{sections/exp-supp}
\input{sections/tech_results}
\input{sections/tech_results_2}

\end{document}

%% file: sections/abstract.tex
We study the problem of missing not at random (MNAR) datasets with binary outcomes. We propose an exponential tilt based approach that bypasses any knowledge on `nonresponse instruments' or `shadow variables' that are usually required for statistical estimation. We establish a sufficient condition for identifiability of tilt parameters  and propose an algorithm to estimate them. Based on these tilt  parameter estimates, we propose importance weighted and doubly robust estimators for any mean functions of interest, and validate their performances in a synthetic dataset. In an experiment with the Waterbirds dataset, we utilize our tilt framework to perform unsupervised transfer learning, when the responses are missing from a target domain of interest, and achieve a prediction performance that is comparable to a gold standard. 



%% file: sections/intro.tex
\section{Introduction}

\label{sec:intro}

Over the years, methods for handling missing data have attracted a lot of attention. The \emph{missing at random} (MAR) assumption, which posits that the missingness mechanism depends solely on observed information, is adequate for completely determining the data distribution and has been rigorously investigated by \citet{rubin1976inference,cheng1994nonparametric,kenward1998likelihood,efromovich2011nonparametric}. In this paper, we're primarily concerned with the \emph{missing not at random} (MNAR) assumption, wherein missingness is allowed to depend on the missing values. Compared to MAR, the MNAR setting is notably more complex, and as argued by \cite{miao2016identifiability} and \cite{wang2014instrumental}, even fully parametric models might not be identified based on the observed dataset's distribution.


The identification under MNAR is sometimes possible with additional tools like nonresponse instruments \citep{wang2014instrumental,wang2021propensity} or shadow variables  \citep{miao2016varieties} coupled with some other assumptions. For example,  \cite{wang2014instrumental} assumes certain likelihood ratios to be monotone and \cite{miao2016varieties} considers certain completeness assumptions. But, till now, most works require knowledge of such non-response instruments or shadow variables to operationalize statistical estimation.

In this paper, we aim to relax this requirement for binary outcomes using an exponential tilt framework.  To provide more context, let us introduce some notation. Suppose a full data-point consists of $(X, Y, R) \sim P$ where $X \in \reals^d$ is a covariate, $Y \in \{0, 1\}$ is a binary response and finally, $R\in \{0, 1\}$ indicates the missingness of $Y$, \ie\  $Y$ is not observed whenever $R = 0$. While the joint distribution of $(X,Y)\mid R = 1$ is fully identified from the distribution of the observed dataset, unfortunately, we cannot identify the joint distribution of $(X, Y) \mid R = 0$ from the observed data distribution under a general MNAR framework. This joint distribution, however, would have been fully identified if we knew the following importance weights: 
\begin{equation}
    \label{eq:imp-weights}
    \textstyle \omega(x, y) = \frac{dP(x\mid Y = y, R = 0) P(Y = y\mid R = 0)}{dP(x\mid Y = y, R = 1) P(Y = y\mid R = 1)}\,.
\end{equation}
As they are typically unknown, our framework seeks to estimate them within the following parametric family defined in terms of a sufficient statistic $t = T(x)$: 
\begin{equation} \label{eq:exponential-tilt}
  \textstyle   \omega(x, y; \theta) = \exp(\alpha_y + {\beta_y }^\top t)\,.
\end{equation} where $\theta= (\alpha_0, \alpha_1, \beta_0, \beta_1)$ are the collection of parameters.  As demonstrated in Lemma \ref{lemma:prop-score}, such a parametric specification allows the missingness propensity score $P(R = 1 \mid X, Y)$ to depend on the outcome $Y$ whenever $\alpha_1 \neq \alpha_0$ or $\beta_1 \neq \beta_0$. Following the nomenclature in \citet{maity2023understanding}, we refer to this parametric  specification in \eqref{eq:exponential-tilt} as the ``exponential tilt model'' and $\theta$ as the ``tilt parameters''.


The term ``exponential tilt'' has previously been used by \cite{kim2011semiparametric}, who considers the following model for the missingness mechanism 
\begin{equation}\label{eq:simpler-tilt}
    \textstyle P(R = 1\mid X, Y) = [1 + \exp(g(X) + \gamma Y ) ]^{-1}
\end{equation}
or, equivalently stated as $\omega(X, Y) =  \exp\{g(X) + \gamma Y\} \cdot P(R = 1) \cdot\{P(R = 0)\}^{-1}$; the $\gamma$ is termed as the tilt parameter, and $g(x)$ is an unknown function of $x$. To facilitate statistical inference, they assume that either they know $\gamma$ or an estimate is available from a different survey dataset.  In fact \citet{wang2014instrumental,miao2016identifiability} argues that such a model is not identifiable from the distribution of the observed dataset if $\gamma$ is unknown. To make it identifiable \citet{wang2014instrumental} introduces an instrumental variable approach: the full covariate $X$ can be split as $X = (U, Z)$ and the unknown function $g$ is only a function of $U$, \ie\ $g(X) = g(U)$. 
 The covariate $Z$, termed as non-response instrument, doesn't affect the missingness mechanism but affects the outcome. 
Such a model has two notable limitations: (1) it doesn't allow interaction terms between $Y$ and $X$ within the missingness propensity score, and (2) the covariate $U$, which satisfies $P(R = 1\mid X, Y) = P(R = 1\mid U, Y)$ and thus we refer to as ``missingness controlling variable'' (MCV), has to be known beforehand. While the shadow variable framework in \citet{miao2016varieties}, described as $Z \perp R \mid (Y, U)$, where $Z$ is called the shadow variable, can allow for an arbitrary interaction between $U$ and $Y$ within the missingness mechanism, the $U$ needs to be known beforehand. In both the frameworks, one needs to make an assumption, $g(X) = g(U)$ or $Z \perp R \mid (Y, U)$, that is not verifiable from the observed dataset. Our model mitigates both of these problems, as 
\begin{enumerate}
    \item the tilt parameters themselves can arbitrarily depend on $Y$, which allows interactions between $X$ and $Y$, and
    \item  it doesn't require any knowledge of instruments or shadow variables as a pre-requirement; rather, it  estimates the importance weights in \eqref{eq:imp-weights} as exponential tilts, thus allowing for a greater flexibility.
\end{enumerate}
Another advantage of our model is discussed in Section \ref{sec:dist-shift}, where we draw a connection with the distribution shift literature and argue that such an interaction between $X$ and $Y$ can allow for a broad class of shifts between the distributions $(X, Y) \mid R = 0$ and $(X, Y) \mid R  = 1$.

\paragraph{Our contribution:}
Regarding the estimation of tilt parameters $\theta$, they are not identifiable in general. To make them identifiable, we propose a sufficient condition in Assumption \ref{assump:identifiability}, which is, in principle, verifiable from the observed dataset. Given the identifiability, in Section \ref{sec:estimation} we develop a constrained optimization framework \eqref{eq:optimization-v1} to estimate the $\theta$. Our optimization relies on matching certain observable marginal distributions with Kullback-Leibler divergence. In the same section, we use the tilt parameter estimates to construct importance weighted estimators for both of the following mean functions $\Ex[\tau(X, Y)]$ and $\Ex[\tau(X, Y)\mid R = 0]$. While the importance weighted estimates are easy to comprehend, they can be unstable, as suggested by \cite{li2023instability} and also seen in our simulation experiments (see Figures \ref{fig:well-sp-simulation} and \ref{fig:missp-simulation}). As an alternative, we also introduce doubly robust estimators and verify their double robustness property under certain misspecifications in our tilt model (see our discussion following Theorem \ref{thm:double-robustness}). 
Then, we validate the performances of these two estimators in a synthetic experiment in Section \ref{sec:simulation}. In Section \ref{sec:waterbirds} we utilize the exponential tilt framework to perform an unsupervised transfer learning task and mitigate certain spurious correlations in the Waterbirds dataset \citep{Sagawa2020Distributionally}.

As a byproduct, the exponential tilt model \eqref{eq:exponential-tilt} estimates  a proxy for the missingness controlling variable (MCV). The model implies that the missingness mechanism
\[
\textstyle P(R = 1\mid X = x, Y = y) =  \textstyle \frac{P(R = 1)}{P(R = 1) + \omega(x, y; \theta ) P(R = 0)}\,,
\] as established in Lemma \ref{lemma:propensity-score},
depends on $X$ only through the $U = (\beta_0^\top T(X), \beta_1 ^\top T(X))$. Thus, $U$ can be treated as MCV. 
One can draw its obvious connections with the assumption of non-response instruments, as it holds $P(R = 1\mid X, Y) = P(R = 1\mid U, Y)$ and with the shadow variable assumption, as it implies $X \perp R \mid Y, U$ (see Lemma \ref{lemma:shadow-variable}).  As such, by estimating the tilt parameters we estimate a proxy for MCV. At the final paragraph of Section \ref{sec:waterbirds} we validate the quality of the estimated MCV in the Waterbirds dataset. 



\subsection{Connection with unsupervised transfer learning} \label{sec:dist-shift}

The study of missing data has a close connection with unsupervised transfer learning. In fact, the unsupervised transfer learning can be treated as a special case of the missing data problem. To understand this connection, let's consider a prediction task with $(X, Y)$; we want to predict the response $Y$ using covariates $X$. In unsupervised transfer learning, one has a labeled sample consisting of both $X$ and $Y$ from a source domain, but only an unlabeled dataset is available in a target domain, \ie\ the $X$'s are observed but $Y$'s  are missing. This setup can be readily put into a missing data framework if we let $R = 1$ for the source dataset and $R = 0$ for the target dataset; then $R$ indicates the missingness of $Y$. 
As such, the full data points consist of the triplets $(X, Y, R)$.  Within unsupervised transfer learning, the goal is to make predictions in the target domain. We consider such a task in Section \ref{sec:waterbirds} with the Waterbirds dataset.

To learn a prediction model on the target domain, one may consider a risk minimization framework
\[
\textstyle \min_{f \in \cF} \Ex[\ell(X, Y; f) \mid R = 0]
\] where $\cF$ is a suitable class of prediction models and $\ell$ is an appropriate prediction loss. For the risk minimization, an estimate of $\Ex[\ell(X, Y; f) \mid R = 0]$ is required. Similar to the MNAR framework, it's impossible to estimate the $\Ex[\ell(X, Y; f) \mid R = 0]$ and fit a prediction model under a general distribution shift between $(X, Y) \mid R = 0$ and $(X, Y) \mid R  = 1$. Thus, typically one makes further assumptions about the distribution shift.

Various distribution shift assumptions are considered. Perhaps the simplest one, the ``covariate shift'' \citep{sugiyama2007covariate,bickel2009discriminative,nair2019covariate},  assumes that $P(Y \mid X, R = 1) = P(Y \mid X, R = 0)$, which is equivalently stated as $Y \perp R \mid X$ and is identical to the missingness at random (MAR) assumption. Another well-studied choice of distribution shift is the ``label shift'' \citep{lipton2018detecting,garg2020unified,maity2022minimax}, which assumes that $P(X \mid Y, R = 1) = P(X \mid Y , R = 0)$, \ie\ $X \perp R \mid Y$. In this case, the missingness controlling variable (MCV) is an empty random variable, and the importance weight in eq. \eqref{eq:imp-weights} simplifies as the following, only depends on $Y$:
\begin{equation} \label{eq:label-shift-IW}
    \textstyle \omega(x, y) = \frac{P(Y = y\mid R = 0)}{P(Y = y\mid R = 1)} =  \frac{P(Y = 0\mid R = 0)}{P(Y = 0\mid R = 1)} + \big \{ \frac{P(Y = 1\mid R = 0)}{P(Y = 1\mid R = 1)} - \frac{P(Y = 0\mid R = 0)}{P(Y = 0\mid R = 1)}\big\} \times  y = \alpha_0 + \gamma y  \,,
\end{equation}
for appropriate $\gamma$ and $\alpha_0$. Thus, it is a special case of the exponential tilt model of \citet{kim2011semiparametric} with a constant $g(x)$ and our exponential tilt model  with a constant $T(x)$.



A more general form of distribution shift, called ``posterior drift'' \citep{scott2019generalized,cai2021transfer,maity2024linear}, allows $P(Y \mid X, R = 1) \neq P(Y \mid X, R = 0)$. Similar to a general MNAR, under such a minimal distribution shift assumption, it is impossible to perform a transfer learning task without making further assumptions; so one naturally assumes how different the  $P(Y \mid X, R = 1) \neq P(Y \mid X, R = 0)$ can be. To keep the discussion relevant to this paper, let's consider our binary outcome case, \ie\ $Y \in \{0, 1\}$. For binary classification, \citet{cai2021transfer} assumes that the Bayes classifiers are identical between source and target, \ie\  $\bbI\{P(Y = 1 \mid X, R = 1) \ge \nicefrac12\}  = \bbI\{ P(Y= 1 \mid X, R = 0) \ge \nicefrac 12\}$, which offers a limited distribution shift.
Our exponential tilt model \eqref{eq:exponential-tilt} parameterizes the differences $P(Y \mid X, R = 1) $ and $ P(Y \mid X, R = 0)$ through the importance weights,  as implied by the following equality that is established in Lemma \ref{lemma:tech2}:
\[
\textstyle \text{logit} \{ P(Y = 1 \mid X, R = 0) \} - \text{logit} \{ P(Y = 1 \mid X, R = 1) \} = \log \omega(X, 1) - \log \omega(X, 0)
\] where for $0 < t< 1$ the $\text{logit}(t) = \log t - \log (1 - t)$.  The exponential tilt assumption of \citet{kim2011semiparametric} offers only a constant difference, as  $\omega(X, Y) =  \exp\{g(X) + \gamma Y\} \cdot P(R = 1) \cdot\{P(R = 0)\}^{-1}$ implies $\log \omega(X, 1) - \log \omega(X, 0) = \gamma$. The label shift also offers a constant difference, as it follows $\log \omega(X, 1) - \log \omega(X, 0) = \gamma$ from eq. \eqref{eq:label-shift-IW}. Our exponential tilt in eq. \eqref{eq:exponential-tilt} offers a more general difference:
\[
\textstyle \log \omega(X, 1) - \log \omega(X, 0) = (\alpha_1 - \alpha_0 ) + (\beta_1 - \beta_0)^\top T(X)\,.
\] This difference is also similar to the linear adjustment model considered in \citet{maity2024linear}, but they require a few labeled samples from the target domain to perform transfer learning, whereas the target labels are not available in our case. The exponential tilt has also been used by \citet{maity2023understanding} for learning and evaluating a model on an unlabeled target domain. Compared to their work, this paper (1) deals with statistical estimation in a missing data framework that is more general than unsupervised transfer learning, and (2) provides doubly robust estimators for the mean functional which enjoy certain robustness properties against  misspecification in the exponential tilt model (see Theorem \ref{thm:double-robustness}).

\subsection{Related work}

\paragraph{Missing data:} Much attention has been paid to developing statistical methods for missing data. \citet{little1992regression} comprises a review of regression techniques with missing covariates. Some popular methods are expectation–maximization (EM) \citep{dempster1977maximum}, multiple imputation \citep{schenker1988asymptotic,rubin2004multiple}, inverse probability weighting \citep{horvitz1952generalization}, and doubly robust estimates \citep{bang2005doubly,tsiatis2006semiparametric,laan2003unified}. Missing data have also been intensively studied in
a range of modern fields, for example, in matrix completion \citep{jin2022matrix,mao2019matrix}, sparse principal component analysis \citep{zhu2022high} and many more. The book by \citet{little2019statistical} provides a comprehensive treatment of the missing data problems.

\paragraph{Missing not at random (MNAR):} 
To mention some of the earlier works in identifiability under MNAR, \citet{heckman1979sample} considered separate parametric models for the outcome and the missingness process, \citet{little1993pattern,little1994class} considered a pattern-mixture parametrization for the incomplete data, \citet{fay1986causal} and \citet{ma2003identification} considered graphical models. \citet{rotnitzky1998semiparametric} and \citet{robins2000sensitivity} investigated the sensitivity in violation of a pre-specified missingness mechanism assumption. \citet{das2003nonparametric,tchetgen2017general,sun2018semiparametric} and \citet{liu2020identification} proposed identification conditions for nonparametric and semiparametric regression models with the help of an instrument variable. \cite{miao2016varieties,miao2024identification} utilizes a shadow variable to make inference on MNAR data. With a further parametrization for $g(x)$ in the tilt model in eq. \eqref{eq:simpler-tilt}, \citet{yu2018estimation,hu2024receiver,liu2022full} considers statistical inference on MNAR data. Finally, we end with a review work by \citet{tang2018statistical} that compiles a more comprehensive review for statistical inference under MNAR.



%% file: sections/setup.tex
\section{Notation and identification}
\label{sec:setup}

We start with some notation. The  full dataset $\{(X_i, Y_i, R_i)\}$ is distributed as $(X_i, Y_i, R_i)\sim \iid \ P$ and the observed dataset is expressed as $\{X_i, Y_iR_i, R_i)\}_{i = 1}^n$, \ie\ $Y_i$'s are missing and set to a default value zero whenever $R_i = 0$.  We assume that the conditional distributions $P(X \mid Y= y, R = r)$ have densities, denoted as $p(x \mid Y = y, R = r)$. We shall denote the $p(x \mid Y = y, R = r) P(Y = y\mid R = r)$ as $p(x,  y\mid R = r)$. With an overload of notation, we shall denote $T(x)$ as $t$ and $T(X)$ as $T$. Then the exponential tilt assumption \eqref{eq:exponential-tilt}  is restated as 
\begin{equation}
    \label{eq:exponential-tilt-v2}
    p(x, Y = y\mid R = 0) = \exp(\alpha_y + {\beta_y}^\top t) p(x, Y = y\mid R = 1)\,.
\end{equation} Before we estimate the tilt parameters $\theta = (\alpha_0, \alpha_1, \beta_0, \beta_1)$, we need to ensure that they are uniquely identified from the observed dataset's distribution. Unfortunately, without additional assumptions, they are not uniquely identified. For a simple counterexample, consider $p(x \mid Y = y , R = r) \sim \cN(\mu_{r, y}, \bbI)$. Since $Y$ are missing for $R = 0$ subpopulation, one only observes a sample from the marginal density 
\[
\textstyle P(Y = 1\mid R  = 0) \phi(x; \mu_{0, 1}) + P(Y = 0\mid R  = 0) \phi(x; \mu_{0, 0})
\] where $\phi(x; \mu)$ is the density of $\cN(\mu, \bbI)$. This marginal can be re-expressed in two ways
\[
\begin{aligned}
    & \textstyle \kappa_1 \phi(x; \mu_{1, 1}) \exp(x^\top (\mu_{0, 1} - \mu_{1, 1})) + \kappa_2 \phi(x; \mu_{1, 0}) \exp(x^\top (\mu_{0, 0} - \mu_{1, 0}))\\
& = \textstyle  \kappa_1'  \phi(x; \mu_{1, 1}) \exp(x^\top (\mu_{0, 0} - \mu_{1, 1})) + \kappa_2' \phi(x; \mu_{1, 0}) \exp(x^\top (\mu_{0, 1} - \mu_{1, 0})) 
\end{aligned}
\] for some $\kappa_1, \kappa_2, \kappa_1', \kappa_2'>0$,  leading to two different possible expressions as 
\begin{equation} \label{eq:marginal}
    \textstyle p(x\mid R = 0) = \sum_{y \in\{0, 1\}} \exp(\alpha_y + {\beta_y}^\top t) p(x, y\mid R = 1)
\end{equation}
with $t = T(x) = x $ and two different values for the tilt parameters. 
Thus, without further assumptions, one may not distinguish between the two tilt parameter values only from the marginal $p(x\mid R = 0)$. 



To have identifiability of the tilt parameters, we consider the following restriction, which ensures a unique solution for $\theta$ from the equation \eqref{eq:marginal}. 
\begin{assumption}[Sufficient condition for identifiability]\label{assump:identifiability}
    For any $\mu_0, \mu_1, \mu_2 \in \reals$ and $\delta_0, \delta_1, \delta_2 \in \reals^d$, where $d = \text{dim}(t)$, it holds 
    \[
    \textstyle \log \big \{\frac{P(Y = 1\mid X = x,  R = 1)}{P(Y = 0\mid X = x,  R = 1)}\big\} \neq \mu_2 + \delta_2^\top t + \log\big\{\frac{e^{\mu_1 + \delta_1 ^\top t} - 1}{e^{\mu_0 + \delta_0^\top t } - 1}\big\}\,. 
    \]
\end{assumption}
As we show in Lemma \ref{lemma:uniqueness}, the above assumption is sufficient to make the tilt parameters $\theta$ uniquely identifiable from \eqref{eq:marginal}.  In principle, this condition can be verified from the observed dataset. For example, if the left-hand side is a quadratic function in $t$, \ie\ if $(X, Y) \mid R = 1$ follows a logistic regression model with $P(Y = 1\mid X = x , R = 1 ) = \{1 + \exp(\alpha_0 + \alpha_1 ^\top t + t^\top A t)\}^{-1}$ for some $\alpha_0 \in \reals$, $\alpha_1 \in \reals^{d}$ and $A \in \reals^{d \times d}$ where $d = \text{dim}(t)$, then the above condition is satisfied. 
Say that the tilt parameter is uniquely identified at $\theta^\star = (\alpha_0^\star, \alpha_1^\star, \beta_0^\star, \beta_1 ^\star)$. In the next section, we estimate this $\theta^\star$, which is referred to as the true  values of the tilt parameters.

%% file: sections/estimation.tex
\section{Estimation}
\label{sec:estimation}

\subsection{Estimation of the tilt parameters}


For estimating the tilt parameters, we consider a distribution matching approach with Kullback-Leibler (KL) divergence. Under the exponential tilt model, the $p(x\mid R = 0)$ is expressed as  in eq. \eqref{eq:marginal}.
Thus, at the population level, we estimate $\theta$ by minimizing the KL divergence 
\[
\textstyle \text{KL} \big(p(x\mid R = 0); \exp(\alpha_1 + \beta_1 ^\top t) p(x ,  y = 1\mid  R = 1) + \exp(\alpha_0 + \beta_0 ^\top t) p(x ,  y = 0\mid  R = 1)\big)
\] subject to the constraint that the right-hand side of eq. \eqref{eq:marginal} integrates to one.  The choice of KL divergence for distribution matching is motivated by its connections to maximum likelihood estimation. The optimization of KL divergence simplifies to
\begin{equation} \label{eq:optimization}
    \begin{aligned}
    & \textstyle \min\limits_{\theta}  \Ex  \big[ - \log\big\{  {e^{\alpha_1 + \beta_1 ^\top T} \eta_1(X) + e^{\alpha_0 + \beta_0 ^\top T} (1 - \eta_1(X))}\big\}  \bigm\vert R = 0\big]  \\
    & \textstyle \quad \text{subject to } \Ex\big[ {e^{\alpha_1 + \beta_1 ^\top T} \eta_1(X) + e^{\alpha_0 + \beta_0 ^\top T} (1 - \eta_1(X))}\bigm\vert R = 1\big] = 1
\end{aligned}
\end{equation} where $T = T(X)$ and $\eta_1(x) = P(Y = 1\mid X = x, R = 1)$, as seen in Lemma \ref{lemma:KL-simplification}, and its empirical version  is:
\begin{equation} \label{eq:optimization-v1}
    \begin{aligned}
    & \textstyle \min\limits_{\theta} \frac1{n_0} \sum_{i = 1}^n - (1 - R_i)\log\big\{  {e^{\alpha_1 + \beta_1 ^\top T_i} \widehat \eta_1(X_i) + e^{\alpha_0 + \beta_0 ^\top T_i} (1 - \widehat \eta_1(X_i))}\big\}    \\
    & \textstyle \quad \text{subject to } \frac1{n_1} \sum_{i = 1}^n  R_i\big[{e^{\alpha_1 + \beta_1 ^\top T_i} \widehat \eta_1(X_i) + e^{\alpha_0 + \beta_0 ^\top T_i} (1 - \widehat \eta_1(X_i))}\big] = 1
\end{aligned}
\end{equation} where $T_i = T(X_i)$, $n_1 = \sum_{i = 1}^n R_i$, $n_0 = n - n_1$ and   $\widehat \eta_1(x) = \eta(x; \widehat \xi)$ is an estimator of $P(Y = 1\mid X = x, R = 1)$ fitted within a model family $\{\eta_1(x; \xi): \xi \in \Xi\}$ on the sample points $(X_i, Y_i)$ with $R_i = 1$.  We estimate the tilt parameter by minimizing this empirical version in eq. \eqref{eq:optimization-v1}.


The optimization in eq. \eqref{eq:optimization-v1} is a constrained optimization, which is solved using an exponentiated gradient  \citep{kivinen1997exponentiated} based  algorithm. To provide an intuition toward the algorithm, consider the Lagrangian version of optimization \eqref{eq:optimization-v1}:
\begin{equation}\label{eq:optimization-v2}
    \begin{aligned}
        & \textstyle \min\limits_{\theta} \max\limits_{\lambda_1, \lambda_2 \ge 0} f_n(\theta ) + \lambda_1 g_n(\theta) - \lambda_2 g_n(\theta), ~~ \text{where}\\
        & \textstyle f_n(\theta) = - \frac1{n_0} \sum_{i = 1}^n  (1 - R_i)\log\big\{  {e^{\alpha_1 + \beta_1 ^\top T_i} \widehat \eta_1(X_i) + e^{\alpha_0 + \beta_0 ^\top T_i} (1 - \widehat \eta_1(X_i))}\big\}\\
    & \textstyle g_n(\theta) =  \frac1{n_1} \sum_{i = 1}^n  R_i\big[{e^{\alpha_1 + \beta_1 ^\top T_i} \widehat \eta_1(X_i) + e^{\alpha_0 + \beta_0 ^\top T_i} (1 - \widehat \eta_1(X_i))}\big] - 1
    \end{aligned}
\end{equation}
The constraint $g_n(\theta) = 0$ is enforced by the inner maximization. If $g_n(\theta) > 0$ for a $\theta$ then the inner maximization is achieved at $\lambda_1 = \infty$ and $\lambda_2 = 0$, which leads to the infinite objective value and makes the $\theta$ infeasible. With a similar logic, a $\theta$ with $g_n(\theta) < 0$ is also infeasible. Thus, the inner minimization safeguards from violating the constraint, as all the feasible solutions must satisfy the $g_n(\theta) = 0$ constraint. From this observation,  the intuition dictates that the $\lambda_1$ should be increased when $g_n(\theta) >0$ and $\lambda_2$ should be increased when $g_n(\theta) <0$ while performing the gradient update for $\theta$. This intuition is at the core of iterations steps of our algorithm, which is described in Algorithm \ref{alg:exp-grad}.

\begin{algorithm}

    \caption{Exponentiated gradient}
    \begin{algorithmic}[1]
        \STATE \textbf{Input:} initial values $\widehat \theta^{(0)},  \eta_1^{(0)} =0, \eta_2^{(0)} = 0$; learning rates $\rho_1, \rho_2 > 0$; Lagrangian multiplier bound $B>0$; constraint violation threshold $\eps$
        \STATE \textbf{Output:} Final tilt parameter estimates $\widehat\theta$
        \FOR{$t = 0, 1, \dots $ }
        \STATE  $\lambda_j^{(t)} \gets \frac{B\exp(\eta_j^{(t)})}{ \exp(\eta_1^{(t)}) + \exp(\eta_2^{(t)})}$ for $j = 1, 2 $
        
        \STATE $\Lambda (\theta) = f_n(\theta) + (\lambda_1^{(t)} - \lambda_2^{(t)}) g_n(\theta)$
        \COMMENTV{Lagrangian objective}
        \STATE $\theta^{(t+1)} \gets \theta^{(t)} - \rho_1 \partial_\theta \Lambda(\theta^{(t)})$
        \COMMENTV{gradient update}
        \STATE $\eta_1^{(t+1)} \gets \eta_1^{(t)} + \rho_2 \log \{g_n(\theta^{(t)}) + 1\}  \bbI\{g_n(\theta^{(t)}) >  \eps\} $ 
        \STATE \COMMENTV{update $\eta_1$} 
        \STATE $\eta_2^{(t+1)} \gets \eta_2^{(t)} - \rho_2 \log \{g_n(\theta^{(t)}) + 1\}   \bbI\{  g_n(\theta^{(t)} <  - \eps \}$
        \STATE \COMMENTV{update $\eta_2$}
        \STATE until convergence. 
        \ENDFOR
    \end{algorithmic}
    \label{alg:exp-grad}
\end{algorithm}

\paragraph{Exponentiated gradient algorithm:} Within the Algorithm \ref{alg:exp-grad}, the Lagrangian multipliers $\lambda_1$ and $\lambda_2$ are expressed
\[
\textstyle \lambda_1 = \frac{B\exp(\eta_1)}{\exp(\eta_1) + \exp(\eta_2)}, ~~ \lambda_2 = \frac{B\exp(\eta_2)}{\exp(\eta_1) + \exp(\eta_2)}
\]
through some other parameters $\eta_1$ and $\eta_2$, for which $\lambda_j$ increases whenever $\eta_j$ increases. Such an expression controls the magnitude of the Lagrange multipliers $\lambda_1 - \lambda_2$ within the optimization \eqref{eq:optimization-v2} through a hyperparameter $B > 0$, as:
\[
\textstyle \lambda_1 - \lambda_2= B \cdot \frac{\exp(\eta_1) - \exp(\eta_2)}{\exp(\eta_1) + \exp(\eta_2)} \in [-B , B]\,. 
\]  Controlling the magnitude of $\lambda_1 - \lambda_2$ offers certain numerical stability in the iterations. The  hyperparameter $\eps$ provides flexibility to enforce the constraint in a weak sense, \ie\ up to $|g_n(\theta)| \le \eps$ and the exact constraint is enforced at $\eps = 0$. To enforce the weak constraint, $\eta_1^{(t)}$ is increased by $\rho_2 \log \{g_n(\theta^{(t)}) + 1\}$ only when $g_n(\theta^{(t)}) > \eps$ and $\eta_2^{(t)}$ is increased by $-\rho_2 \log \{g_n(\theta^{(t)}) + 1\}$ only when $g_n(\theta^{(t)}) <- \eps$. Finally, the algorithm is terminated when the following convergence condition is satisfied:
\[
\textstyle \frac{\|\widehat \theta^{(t+1)} - \widehat \theta ^{(t)}\|}{\| \widehat \theta ^{(t)}\|} + \max\big\{|g(\widehat \theta^{(t)})| - \eps, 0\big\} \le \text{a convergence threshold value}.
\] Further implementation details about the algorithm are deferred to the Appendix \ref{sec:exp-supp}.

Having estimated the tilt parameters using the Algorithm \ref{alg:exp-grad} as $\widehat \theta = (\widehat \alpha_0, \widehat \alpha_1, \widehat \beta_0, \widehat \beta_1)$, we  estimate the importance weights simply as 
\begin{equation}\label{eq:imp-weights-estimates}
    \textstyle \widehat \omega(x, y)  = \exp(\widehat \alpha_y + \widehat \beta_y^\top t) = \omega(x, y; \widehat \theta)\,, 
\end{equation}

\paragraph{Empirical likelihood approach:} As an alternative to the distribution matching and Algorithm \ref{alg:exp-grad}, in Appendix \ref{sec:empirical-likelihood} we also consider an empirical likelihood approach to estimate the tilt parameters. But a comparison of empirical results finds that the exponentiated gradient approach outperforms this empirical likelihood approach. The remaining details are  deferred to the Appendix \ref{sec:empirical-likelihood}.

\subsection{Estimation of mean functional}
Now, let's discuss estimating the two following mean functions: 
\begin{equation}\label{eq:mean-functions-1}
    \mu_0 = \Ex[\tau(X, Y) \mid R = 0], ~~ \text{and} ~~ \mu= \Ex[\tau(X, Y)]\,.
\end{equation}
Estimating these two mean functions is sufficient to estimate any general mean function, as for any $\Psi(X, Y, R)$ it follows:
\[
\begin{aligned}
    \Ex[\Psi(X, Y, R)] & = \Ex[R\cdot \Psi(X, Y, 1) + (1 - R)\cdot  \Psi(X, Y, 0) ] \\
    & = \Ex[\Psi(X, Y, 1) + (1 - R) \{\Psi(X, Y, 0) - \Psi(X, Y, 1)\} ]\\
    & = \Ex [\Psi(X, Y, 1)] + P(R = 0) \cdot \Ex[\Psi(X, Y, 0) - \Psi(X, Y, 1) \mid R = 0]\,,
\end{aligned}
\] which can be readily estimated by plugging in estimates for $\Ex [\Psi(X, Y, 1)]$, $  \Ex[\Psi(X, Y, 0) - \Psi(X, Y, 1) \mid R = 0]$ and $P(R = 0)$.

To estimate the two mean functions in \eqref{eq:mean-functions-1} we utilize the estimated importance weights from \eqref{eq:imp-weights-estimates} and calculate the following importance weighted (IW) estimates
\begin{equation}
    \textstyle \widehat \mu_{0, \text{IW}} = \frac1{n_1} \sum_{i = 1}^n R_i \widehat \omega(X_i, Y_i) \tau(X_i, Y_i), ~~ \widehat \mu_{ \text{IW}} = \frac1{n} \sum_{i = 1}^n R_i \{1 + \frac{1 - \widehat \pi_r}{\widehat \pi_r}\widehat \omega(X_i, Y_i)\}  \tau(X_i, Y_i)\,, 
\end{equation} where $\widehat \pi_r = \nicefrac{n_1}{n}$ is an estimate of $P(R = 1)$. 
The intuition of these estimates lies in expressing the mean functions using importance weights: 
\[
\begin{aligned}
    \textstyle \mu_0 & = \Ex[\tau(X, Y) \mid R = 0] = \Ex[\tau(X, Y) \omega(X, Y) \mid R = 1]\,\\
    \textstyle \mu = \Ex[\tau(X, Y) ] & = P(R = 1 )\cdot \Ex[\tau(X, Y)  \mid R = 1] + P(R = 0)\cdot\Ex[\tau(X, Y)  \mid R = 0]\\
    & = \textstyle P(R = 1 ) \cdot\Ex\big[\tau(X, Y) \big\{ 1+ \frac{P(R = 0)}{P(R = 1)} \omega(X, Y) \big\}\big \vert  R = 1\big] 
\end{aligned}
\] where we replace the $\omega(X, Y)$ and $P(R = 1)$ with their estimates $\widehat \omega(X, Y)$ and $\widehat \pi_r$ to obtain the IW estimates.
We also note that the IW estimate for $\mu = \Ex [\tau(X, Y)]$ is identical to the inverse probability weighted (IPW) estimate, as it follows from the next lemma. 
\begin{lemma} \label{lemma:propensity-score}
    The (non-)missingness propensity score is
        \begin{equation} \label{eq:prop-score}
            \begin{aligned}
                 P(R = 1\mid X = x, Y = y) = & \textstyle \frac{P(R = 1)}{P(R = 1) + \omega(x, y) P(R = 0)}\,,
            \end{aligned}
        \end{equation} and it follows: 
        \[
        \begin{aligned}
          \textstyle    \frac{R }{ P(R = 1\mid X , Y)}  = R \big \{ 1 + \frac{P(R = 0)}{P(R = 1)} \omega(X, Y) \big \} \,.
        \end{aligned}
        \]
\end{lemma}
It is straightforward to see from the above lemma that plugging in $\widehat \pi$ and $\widehat \omega(x, y)$ as the estimates of $P(R = 1)$ and $\omega(x, y)$ within \eqref{eq:prop-score} the IW estimate is identical to the IPW estimate \[
\textstyle \widehat \mu_{ \text{IPW}} = \frac{1}{n} \sum_{i = 1}^n \frac{R_i \tau(X_i, Y_i)}{\widehat \eta_r(X_i, Y_i)}  = \frac{1}{n} \sum_{i = 1}^n {R_i \big \{ 1 + \frac{1- \widehat \pi_r}{\widehat \pi_r} \widehat \omega(X_i, Y_i) \big \} \tau(X_i, Y_i)}\,,
\] where $\widehat \eta_r(x, y) = \big \{ 1 + \frac{P(R = 0)}{P(R = 1)} \omega(X, Y) \big \}^{-1}$ is an estimate of $P(R = 1\mid X = x, Y = y)$. 

The consistency of IW estimators relies on consistently estimating the importance weights $\omega(x, y)$, as established in the first part of Lemma \ref{lemma:DR-bias}. While the IW  estimates are easy to understand, they have a limitation. They can be unstable, as argued by \cite{li2023instability} and also seen in our simulation experiments (see Figures \ref{fig:well-sp-simulation} and \ref{fig:missp-simulation}) that they tend to have a high variance. As an alternative, we propose double robust (DR) estimators. The DR estimators generally require an outcome regression estimate, which,  in our case, is an estimate of 
\[
\textstyle \Ex[\tau(X, Y) \mid X , R = 0] =  \tau(X, 0) + \Ex[Y \mid R = 0, X] \{\tau(X, 1) - \tau(X, 0)\}\,.
\]
The $\tau(X, 0)$ and $\tau(X, 1)$ are known, and thus, we only need to estimate $\Ex[Y \mid R = 0, X] $.  With the definition of a log odds ratio: 
\begin{equation}
    \textstyle \gamma(y \mid x) = \log \big \{ \frac{P(R = 1\mid Y = y, X = x) P(R = 0\mid Y = 0, X = x)}{P(R = 0\mid Y = y, X = x) P(R = 0\mid Y = 1, X = x)} \big \}\,,
\end{equation}
the following lemma provides an expression for $\Ex[Y \mid X = x, R = 0]$.
\begin{lemma}\label{lemma:conditional-mean-r0}
    The conditional expectation of the missing outcomes is 
    \begin{equation}
        \Ex[Y \mid R = 0, X = x] = \textstyle \frac{\exp(\gamma(1 \mid x))  {P(Y = 1\mid R = 1, X = x)}}{\exp(\gamma(1 \mid x))  {P(Y = 1\mid R = 1, X = x)} + P(Y = 0\mid R = 1, X = x)  }\,.
    \end{equation}
\end{lemma}
Estimating the conditional expectation requires estimating its two components:  $\gamma(1 \mid x)$ and $P(Y = 1 \mid X = x, R =1)$. 
Under the exponential tilt model \eqref{eq:exponential-tilt} $$\textstyle \gamma(1 \mid x) = (\alpha_1^\star - \alpha_0^\star) + (\beta_1^\star - \beta_0^\star)^\top t \triangleq \gamma(1 \mid x; \theta^\star)\,,$$ which can be estimated as 
$
\textstyle \widehat \gamma(1 \mid x) = \gamma(1 \mid x; \widehat \theta)  =  (\widehat \alpha_1 - \widehat \alpha_0) + (\widehat \beta_1 - \widehat \beta_0)^\top t.$
We note a distinction between the true log-odd ratio, denoted as $\gamma(1\mid x)$, and the log odds ratio  under the exponential tilt model, denoted as $\gamma(1\mid x; \theta^\star)$. 
The  estimate $\widehat \eta_1(x) = \eta_1(x; \widehat \xi)$ of $P(Y = 1 \mid X = x, R =1)$ was required within the optimization \eqref{eq:optimization-v1} for estimating the tilt parameters. We can use the same estimates here. We plug in these two estimates to obtain 
\begin{equation}
    \label{eq:or}
    \textstyle \widehat m_0(x) = m_0(x; \widehat \theta, \widehat \xi) = \frac{\exp(\widehat \gamma(1 \mid x))  {\widehat \eta_1(x)}}{\exp(\widehat \gamma(1 \mid x))  {\widehat \eta_1(x)} +  \{1 - \widehat \eta_1(x)\}  }\,.
\end{equation} where the dependence over the  $\widehat \theta$ and $\widehat \xi$ are hidden within the estimates $\widehat \gamma(1\mid x)$ and $\widehat \eta_1(x)$. 
Using this outcome regression estimate, we describe the double robust estimates as
\begin{equation} \label{eq:DR-estimates}
    \begin{aligned}
         & \textstyle \widehat \mu_{0, \text{DR}} = \frac1{n_1} \sum_{i = 1}^n R_i \widehat \omega(X_i, Y_i) \{\tau(X_i, Y_i) - \widehat m_{0, \tau} (X_i)\} + \frac1{n_0} \sum_{i = 1}^n (1 - R_i) \widehat m_{0, \tau}(X_i)\\
    & \textstyle \widehat \mu_{ \text{DR}} = \frac1{n} \sum_{i = 1}^n R_i \{\frac{1 - \widehat \pi_r}{\widehat \pi_r} + \widehat \omega(X_i, Y_i)\}  \{ \tau(X_i, Y_i) - \widehat m_{0, \tau}(X_i)\} +  \frac1{n} \sum_{i = 1}^n \widehat m_{0, \tau}(X_i)\,,
    \end{aligned}
\end{equation} where $\widehat m_{0, \tau}(X_i) = m_{0, \tau}(X_i; \widehat \theta, \widehat \xi) =  \tau(X_i, 0) +  m_0(X_i; \widehat \theta, \widehat \xi) \{\tau(X_i, 1) - \tau(X_i, 0)\}$
estimates the $\Ex[\tau(X_i, Y_i) \mid R_i = 0, X_i]$.

To describe their double robustness property, we consider the following. Suppose, $\widehat \theta$ and $\widehat \xi$ converge in probability to their corresponding population values  $\theta^\star$ and $\xi^\star$. 
Then we calculate the asymptotic bias of both IW and DR estimates in the following lemma. 
\begin{lemma} \label{lemma:DR-bias}
    The asymptotic biases for the IW \eqref{eq:imp-weights-estimates} and DR \eqref{eq:DR-estimates} estimates  are 
    \begin{equation}
        \begin{aligned}
        \text{a-bias} (\widehat \mu_{0, \text{IW}}) & = \textstyle \Ex \big [\big \{\frac{R }{P(R = 1)} \omega(X, Y; \theta^\star) - \frac{1 - R}{P(R = 0)} \big \}  \tau(X, Y)  \big] \\
            \text{a-bias} (\widehat \mu_{ \text{IW}}) & = \textstyle \Ex \big [\big \{\frac{R }{\eta_r (X, Y, \theta^\star)}  - 1  \big \} \tau(X,  Y)  \big] \\
            \text{a-bias} (\widehat \mu_{0, \text{DR}}) & = \textstyle \Ex \big [\big \{\frac{R }{P(R = 1)} \omega(X, Y; \theta^\star) - \frac{1 - R}{P(R = 0)} \big \} \big \{ \tau(X,  Y) -  m_{0, \tau}(X; \theta^\star, \xi^\star)\big \} \big] \\
            \text{a-bias} (\widehat \mu_{ \text{DR}}) & = \textstyle \Ex \big [\big \{\frac{R }{\eta_r (X, Y, \theta^\star)}  - 1  \big \} \big \{ \tau(X,  Y) -  m_{0, \tau}(X; \theta^\star, \xi^\star)\big \} \big] 
        \end{aligned}
    \end{equation} 
    \[
    \begin{aligned}
   \text{where}, ~~ \textstyle m_{0, \tau}(X; \theta^\star, \xi^\star) = \tau(X, 0) + \{\tau(X, 1) - \tau(X, 0)\} m_0(X; \theta^\star, \xi^\star)\\
        \textstyle m_0(X; \theta^\star, \xi^\star) =\frac{\exp( \gamma(1 \mid X; \theta^\star))  { \eta_1(X; \xi^\star)}}{\exp( \gamma(1 \mid X; \theta^\star))  { \eta_1(X; \xi^\star)} +  \{1 - \eta_1(X; \xi^\star)\}  }\\
        \textstyle \gamma(1 \mid X; \theta^\star) = (\alpha_1^\star - \alpha_0^\star) + (\beta_1^\star - \beta_0^\star)^\top T(X)\\
        \textstyle \eta_r (X, Y, \theta^\star) = \frac{P(R = 1)}{P(R = 1) + \omega(X, Y; \theta^\star) P(R = 0)}\,.
    \end{aligned}
    \]
\end{lemma}
With the expressions of asymptotic biases, we're now ready to discuss the asymptotic bias of IW estimates and the double robustness property of DR estimates, which is presented in the following theorem: 
\begin{theorem}[Double robustness]\label{thm:double-robustness}
    The following holds for the IW and DR estimators. 
    \begin{enumerate}
        \item[(i)] If $\Ex[R\mid X, Y] = \eta_r(X, Y; \theta^\star)$ then $\text{a-bias} (\widehat \mu_{0, \text{IW}}) = 0$,  $\text{a-bias} (\widehat \mu_{ \text{IW}}) = 0$, $\text{a-bias} (\widehat \mu_{0, \text{DR}}) = 0$ and $\text{a-bias} (\widehat \mu_{ \text{DR}}) = 0$. 
        \item[(ii)] If $\gamma(1\mid x) = \gamma(1\mid x, \theta^\star)$ and $\Ex[Y \mid X, R = 0] = m_0(X; \theta^\star, \xi^\star)$ then $\text{a-bias} (\widehat \mu_{0, \text{DR}}) = 0$ and $\text{a-bias} (\widehat \mu_{ \text{DR}}) = 0$.
    \end{enumerate}
\end{theorem}
The first part requires well-specification in the exponential tilt model \eqref{eq:exponential-tilt}, as, from Lemma \ref{lemma:propensity-score}, the well-specification of  
\[
\textstyle \eta_r (X, Y, \theta^\star) = \frac{P(R = 1)}{P(R = 1) + \omega(X, Y; \theta^\star) P(R = 0)}\,,
\] is equivalent to the well-specification of $\omega(X, Y; \theta^\star)$.
But, the second part can allow for certain misspecification in the tilt model. Looking at the expression of $\Ex[Y \mid X, R = 0]$ in Lemma \ref{lemma:conditional-mean-r0}, the second part requires the 
$\gamma(1 \mid X; \theta^\star) = (\alpha_1^\star - \alpha_0^\star) + (\beta_1^\star - \beta_0^\star)^\top T(X)$ and $\eta_1(X; \xi^\star)$ to be well-specified. The $\eta_1(X; \xi^\star)$ is a model for $P(Y = 1 \mid X, R = 1)$, so, in principle, one can consider a non-parametric classification for predicting $Y$ using $X$ on $R = 1 $ dataset to achieve well-specification. The well-specification in $\gamma(1 \mid X; \theta^\star)$, on the other-hand, only requires 
\[
\textstyle \log \big \{ \frac{P(R = 1\mid Y , X ) P(R = 0\mid Y = 0, X )}{P(R = 0\mid Y , X ) P(R = 0\mid Y = 1, X )} \big \} = (\alpha_1^\star - \alpha_0^\star) + (\beta_1^\star - \beta_0^\star)^\top T(X)\,.
\]
This can be true even when the exponential tilt assumption \eqref{eq:exponential-tilt} is violated. So, the DR estimators are asymptotically unbiased for those misspecifications in the exponential tilt assumption if the outcome regression model $m_0(X; \theta^\star, \xi^\star)$ is well-specified.

%% file: sections/experiments.tex
\section{Synthetic experiments}
\label{sec:simulation}

In this section, we validate the effectiveness of IW and DR estimates with  a synthetic experiment consisting of normal mixture models. We set  $P(R = 1) = \nicefrac12$. For  $r \in \{0, 1\}$ we denote $P (Y = 1\mid R = r)$ as $\pi_{1 \mid r}$. The data is generated as the following
\begin{equation} \label{eq:sim-well-sp}
    \begin{aligned}
     R \sim \text{ber}(\nicefrac{1}{2}), ~~   Y \mid R = r \sim \text{ber}(\pi_{1 \mid r})  , ~~ X \mid Y = y, R = r \sim \cN_2(\mu_{r, y}, \sigma_y^2 \bbI_2),\\
     \textstyle ~~ \mu_{r, y} = \begin{bmatrix}
         (2y - 1) (1 - 2r)\\
         2 \cdot (2 y - 1)
     \end{bmatrix}
    \end{aligned}
\end{equation} where $\{\mu_{r, y}: r, y \in \{0, 1\}\} \in \reals^ 2$ are group-wise means for $X$.  We set $\pi_{1 \mid 1 } = 0.4$,  $\pi_{1 \mid 0} = 0.6$, $\sigma_0 = 1$ and vary $\sigma_1$. Note that $\omega(x, y)$ follows the exponential tilt model \eqref{eq:exponential-tilt} with a linear function in $x$ ($T(x) = x$), because 
\[
\textstyle \log \omega(x, y) = \log \frac{p(x, y \mid R = 0)}{p(x, y\mid R = 1)} = c_y - \frac{1}{2\sigma_y^2} \big \{ \|x - \mu_{0, y}\|_2^2 - \|x - \mu_{1, y}\|_2^2  \big\} = c'_y + \delta_y^\top  x\,.
\] for some $c_y, c_y' \in \reals$ and $\delta_y \in \reals^2$. 
Also, the log odds for $Y \mid X = x$ within $R = 1$ is 
\begin{equation} \label{eq:log-odds-well-sp}
    \textstyle \log \big \{ \frac{P(Y = 1\mid X = x, R = 1)}{P(Y = 0\mid X = x , R = 1)}\big \} = c + \frac12 \|x - \mu_{1, 0} \|^2 - \frac{1}{2\sigma_1^2 } \|x - \mu_{1, 1} \|^2
\end{equation}
for some $c \in \reals$, which is quadratic in $x$ whenever $\sigma_1 \neq 1$. Thus, the identifiability assumption \ref{assump:identifiability} of the tilt parameters is satisfied. Under this model, we generate $n = 400$ random samples $\{(X_i, Y_i, R_i)\}$ and get the tilt parameter estimates $\widehat \theta$. The estimation of $\theta$ requires a probabilistic classifier $\widehat \eta_1(X)$ that estimates $P(Y = 1\mid X , R = 1 )$. We fit this probabilistic classifier on a separate random sample $\{(X_i', Y_i')\}$ of size $200$, that is drawn from the distribution of $(X, Y) \mid R = 1$. We fit a random forest to learn such a classifier and do not use direct knowledge on the data generating structure of $(X, Y) \mid R = 1$. The estimates $\widehat \theta$ and $\widehat \eta_1$ are then utilized to obtain IW and DR estimates for both $\mu= \Ex[Y] =  (\pi_{1\mid 1} + \pi_{1 \mid 0})/2 = (0.4 + 0.6) / 2 = 0.5$ and $\mu_0 = \Ex[Y \mid R = 0] = \pi_{1\mid r} = 0.6$. 

\begin{figure}[h]
    \centering
    \includegraphics[width=1\linewidth]{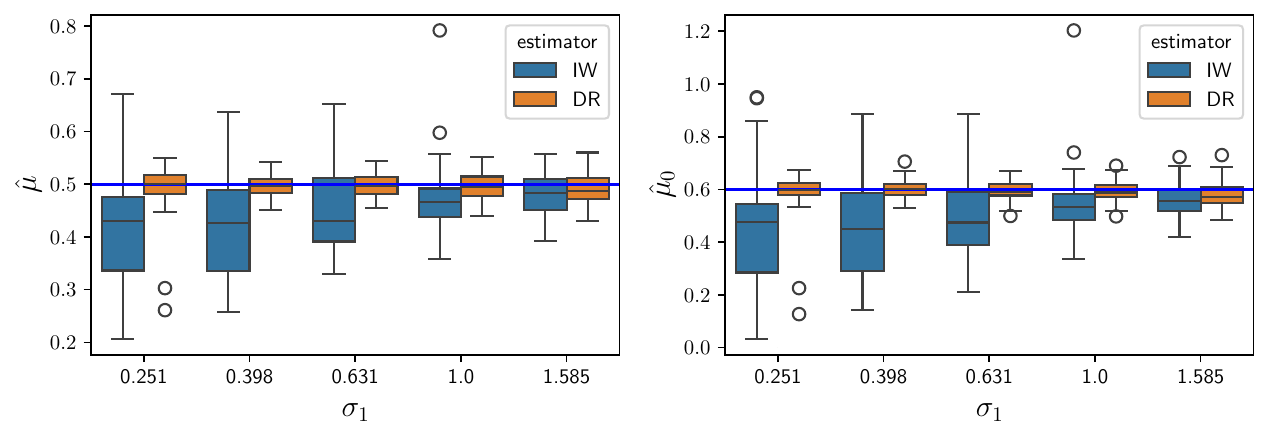}
    \caption{Different estimates in well-specified exponential tilt model in \eqref{eq:sim-well-sp}. The box-plots for both IW and DR estimates for both $\mu$ and $\mu_0$ are presented over multiple values of $\sigma_1$. The blue line indicates the true parameter value. In general, DR estimates are better than IW estimates in terms of both bias and variance.}
    \label{fig:well-sp-simulation}
\end{figure}

The IW and DR estimates for both $\mu$ and $\mu_0$ are presented in Figure \ref{fig:well-sp-simulation} in the form of box-plots over multiple values of $\sigma_1$. As seen in the figure, the DR estimates are generally better than IW estimates both in their biases and variances. 

\paragraph{Misspecification in exponential tilt:}
We also consider a misspecification in the exponential tilt model. The model \begin{equation} \label{eq:sim-missp}
    \begin{aligned}
      R \sim \text{bernoulli}(\nicefrac{1}{2}), ~~    Y \mid R = r \sim \text{bernoulli}(\pi_{1 \mid r})  , ~~ X \mid Y = y, R = r \sim \cN_2(\mu_{r, y}, \sigma_{r, y}^2 \bbI_2)
    \end{aligned}
\end{equation} is mostly similar to before; the only difference is that the  $\var(X \mid Y = y, R = 0) = \sigma^2_{0, y}$ and $\var(X \mid Y = y, R = 1)= \sigma_{1, y}^2$ might be different. In this case
\[
\textstyle \log \omega(x, y) = \log \frac{p(x, y \mid R = 0)}{p(x, y\mid R = 1)} = c_y - \frac{\|x - \mu_{0, y}\|_2^2}{2\sigma_{0, y}^2} + \frac{\|x - \mu_{1, y}\|_2^2}{2\sigma_{1, y}^2} \,
\] might not be a linear function in $x$ if $\sigma_{0, y} \neq \sigma_{1, y}$. Yet, we consider $T(x) = x$, thus  fit a misspecified exponential tilt model. We set $\sigma_{1, 1 } = \sigma_1$, $\sigma_{1, 0} = 1$ $\sigma_{0, 1} = 1$ and $\sigma_{0, 0} = \sigma_1$ to ensure that $\sigma_{0, y} \neq \sigma_{1, y}$ whenever $\sigma_1 \neq 1 $. 

We provide a similar plot for the IW and DR estimates of both $\mu$ and $\mu_0$ in Figure \ref{fig:missp-simulation}. The estimates are worse compared to the well-specified case, which is expected. But, especially for DR estimates, the performances are reasonable, especially when $\sigma_1$ is not far away from one, \ie\ misspecification is not too high. As another expected finding, there's virtually no difference between the results in well-specified and misspecified cases when $\sigma_1 = 1$. In this case, the true $T(x)$ function is indeed linear, \ie\ our exponential tilt model is well-specified.

\begin{figure}[h!]
    \centering
    \includegraphics[width=1\linewidth]{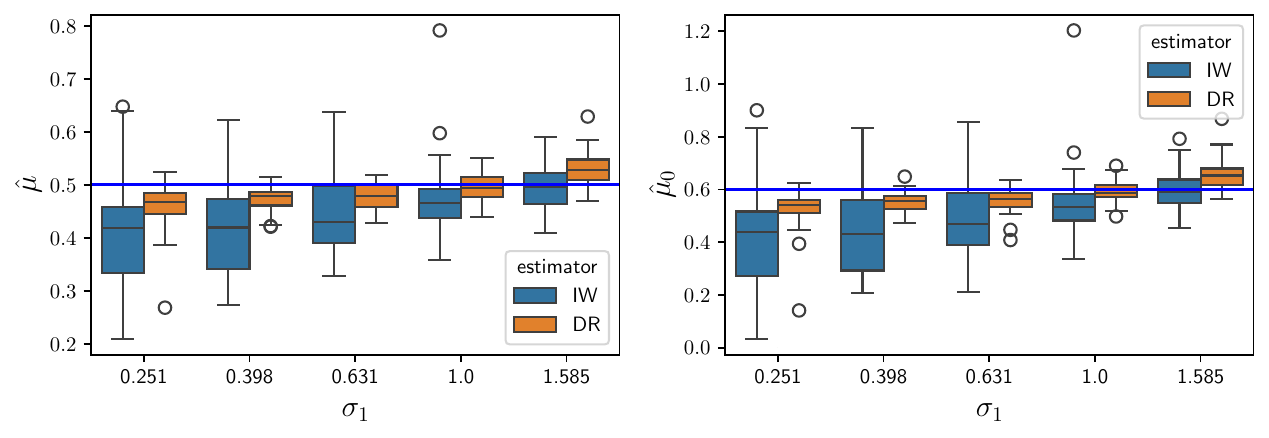}
    \caption{Different estimates in misspecified exponential tilt model in \eqref{eq:sim-missp}. Even in misspecified case DR estimators can provide  reasonable estimates when the misspecification is not too large, \ie, $\sigma_1$ is not far away from one.}
    \label{fig:missp-simulation}
\end{figure}

\section{Unsupervised transfer learning in Waterbirds case-study}
\label{sec:waterbirds}

We demonstrate the efficiency of the exponential tilt model by mitigating spurious correlation within the Waterbirds benchmark dataset\footnote{\url{https://www.kaggle.com/datasets/bahardibaie/waterbird}} \citep{Sagawa2020Distributionally}. The dataset combines bird photographs from the Caltech-UCSD Birds-200-2010 dataset \citep{wah2011caltech} and image backgrounds from the Places dataset \citep{zhou2017places}. The birds are labeled as one of $Y\in$  \{water or landbird\} and placed on one of $B \in $  \{water or land background\}. Hence, the images are divided into four groups: landbirds on land (0), landbirds on water (1), waterbirds on land (2) and waterbirds on water (3). The images are embedded into 512-dimensional vectors using ResNet18 \citep{he2016deep}, which is pre-trained on the Imagenet database \citep{deng2009imagenet}. We use these embedding vectors as covariates.

The data is automatically divided into three splits: train, validation and test. The validation and test splits are designed to be identically distributed, and thus, combine them to create a target domain. We use the train split as the source domain.
The source dataset is designed to be highly imbalanced: in both cases approximately $95\%$ of landbird and waterbird images have their respective land and water backgrounds; see the left plot of Figure \ref{fig:waterbirds-1}. This introduces a spurious correlation between the bird and background types in the source distribution. On the contrary, the target dataset is well balanced: both the landbird and waterbird images have the same proportions of land and water backgrounds. Furthermore, the image distributions within each of the four groups are designed to be the same across source and target domains. A classifier that has been trained on source data will typically depend on the background information because of the imbalance referred to above. As a result, the model can have poor performance on the `landbirds on water' and `waterbirds on land' groups, which will be reflected in the target data performance. Indeed, a penalized logistic regression model that we trained on $4795$ observations from the source dataset has only $79.52\%(\pm 0.07\%)$ prediction accuracy on some 1749 observations from the target dataset. This is substantially lower than the $92.63\%(\pm 0.05\%)$ accuracy on the same 1749 observations as before of a similar model trained on another 5244 observations from the target dataset. We refer to the latter accuracy as the gold standard target accuracy. A challenge for this particular transfer problem is to overcome the spurious correlation between the birds and the background that is present in the source data and improve prediction accuracy on the target domain.

To put the Waterbirds dataset within  missing data framework, we pretend that  the source domain is $R=1$, where the bird types ($Y$) are observed and the target domain is $R=0$ where the bird types ($Y$) are missing. As seen in the left plot of Figure \ref{fig:waterbirds-1}, the proportions of the groups differ between the source and target domains based on both bird type $Y$ and background type $B$, which means this is an instance of a missing not at random (MNAR) response.

This is a case of sub-population shift; the image distribution and hence the distribution of $X$  within $(Y, A)$ groups are designed to be identical across source and target domains, whereas the only difference between the two distributions is in proportions of $(Y, A)$. Expressing the subpopulation shift as $X \perp R \mid Y, A$ we notice that $X$ is the nonresponse instrument \citep{wang2021propensity} or shadow variable \citep{miao2016varieties}. In this case, the tilt weights in \eqref{eq:exponential-tilt} depend only on $Y$ and $A$, \ie\ $P(R = 1 \mid X, Y, A) = P(R = 1\mid Y, A)$, and hence $A$ is the missingness control variable (MCV).  The techniques discussed in \citet{miao2016varieties,wang2021propensity} can be utilized only if $A$ is observed in both source and target datasets. To demonstrate the efficacy of the exponential tilt model in bypassing the requirement on such MCV, we assume that we do not observe the $A$. While the techniques of \citet{miao2016varieties,wang2021propensity} are not applicable in this case, we want to see whether our exponential tilt framework could be useful in making predictions on the target dataset. 

\begin{figure}[h!]
    \centering
    \includegraphics[width=0.45\linewidth]{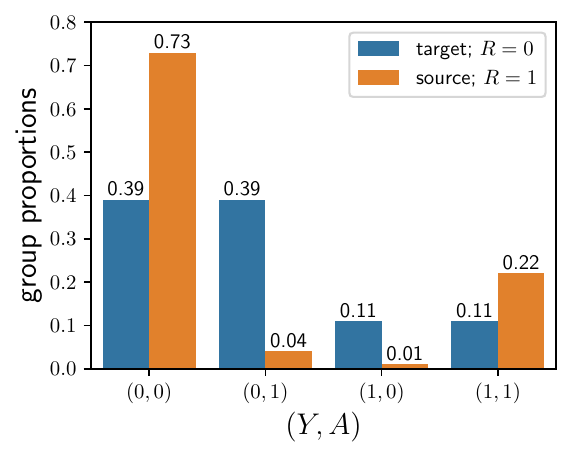}
    \hfill
    \includegraphics[width=0.5\linewidth]{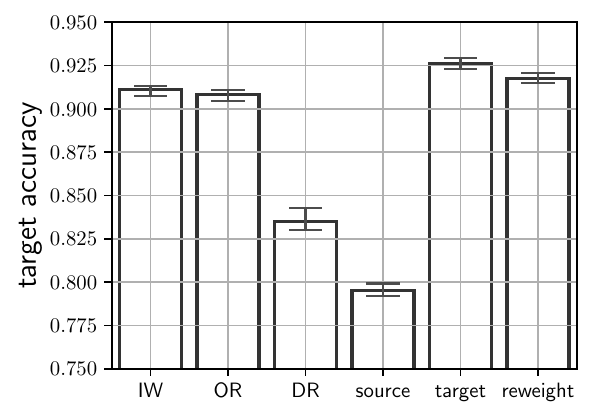}
    \caption{{\bf Waterbirds case-study.} The \emph{left plot} presents proportions of different groups in source and target domain and the \emph{right plot} provides performances of different prediction models on the target domain.}
    \label{fig:waterbirds-1}
\end{figure}

To estimate the tilt weights, we require a probabilistic classifier $P(Y = 1\mid R = 1, X)$. We train a single hidden layer neural network on the source dataset to fit this probabilistic classifier. The functional  structure of  single hidden layer neural network satisfies the Assumption \ref{assump:identifiability}. With this, we estimate the tilt parameters according to Algorithm \ref{alg:exp-grad} with $T(x) = x$. While estimating the tilt parameter, we only use 75\% of the target dataset with 5244 observations, henceforth referred to as ``target-train''. The remaining 25\% with 1749 observations, referred to as ``target-test'', will be used for evaluating prediction performance in the target domain. We then use the tilt weights to fit the following regularized logistic regression models.  
\begin{equation}\label{eq:logistic-IW-OR-DR}
    \begin{aligned}
         \widehat \beta_{\text{IW}} & = \textstyle \argmin_\beta \frac1{n_1} \sum_{i = 1}^n R_i \widehat \omega(X_i, Y_i) \ell(X_i, Y_i; \beta)  + \lambda \|\beta\|_2^2\\
         \widehat \beta_{\text{OR}} & = \textstyle \argmin_\beta \frac1{n_0} \sum_{i = 1}^n (1 - R_i)  \ell(X_i, \widehat m_0(X_i); \beta)  + \lambda \|\beta\|_2^2\\
        \widehat \beta_{\text{DR}} & = \textstyle \argmin_\beta \textstyle \frac{1}{n_1}\sum_{i = 1}^n R_i \widehat \omega(X_i, Y_i) \{\ell(X_i, Y_i; \beta)  - \ell(X_i, \widehat m_0(X_i); \beta) \} \\
       & ~~~~~~~ \quad \textstyle + \frac1{n_0} \sum_{i = 1}^n (1 - R_i) \ell(X_i, \widehat m_0(X_i); \beta)  + \lambda \|\beta\|_2^2
    \end{aligned}
\end{equation}
where $\ell(X, Y; \beta)$ is the loss for logistic regression, defined as $\ell(X, Y; \beta) = - Y (\beta_0 + \beta_1 ^\top X) + \log \{ 1+ \exp(\beta_0 + \beta_1^\top X)\}$,  $\widehat m_0(X_i)$ is as described in equation \eqref{eq:or} and for a $p \in [0, 1]$ the $\ell(X, p; \beta) = p \ell(X, 1; \beta) + (1 - p) \ell(X, 1; \beta)$.

The first and third essentially minimize the IW and DR estimates of $\Ex[\ell(X, Y; \beta) \mid R = 0]$, described in Section \ref{sec:estimation}. The second estimate, which we call the outcome regression (OR) estimate, replaces the missing $Y_i$ in the target-train dataset with their estimates $\widehat m_0(X_i)$. We compare these three logistic regression models against the following benchmarks: 
\begin{enumerate}
    \item {\bf source}, which is a penalized logistic regression model fitted on the source dataset 
\[
\textstyle \widehat \beta_{\text{source}}  = \textstyle \argmin_\beta \frac1{n_1} \sum_{i = 1}^n R_i  \ell(X_i, Y_i; \beta)  + \lambda \|\beta\|_2^2
\]
\item  {\bf target}, an oracle benchmark, that fits a penalized logistic regression model on the target-train split with the true $Y_i$'s 
\[
\textstyle \widehat \beta_{\text{target}}  = \textstyle \argmin_\beta \frac1{n_0} \sum_{i = 1}^n (1 - R_i)  \ell(X_i, Y_i; \beta)  + \lambda \|\beta\|_2^2
\]
\item and finally, {\bf reweighting}, another oracle benchmark, that fits a penalized logistic regression model fitted on the source split with the true importance weights:
\[
\textstyle \widehat \beta_{\text{reweight}}  = \textstyle \argmin_\beta \frac1{n_1} \sum_{i = 1}^n R_i \omega^\star(A_i, Y_i)  \ell(X_i, Y_i; \beta)  + \lambda \|\beta\|_2^2
\] Note that the true importance weights $\omega^\star(A, Y)$ depends on $A$, so in this case we assume that we observe $A$ on the source data. 
\end{enumerate}

The test accuracies of these six logistic regression models are calculated on the remaining 25\% of the ``target-test'' split in the target dataset. The whole procedure is repeated 20 times over different random train-test splits of the target dataset. The accuracies, along with their error-bars, are presented in the right plot of Figure \ref{fig:waterbirds-1}. Since the `target' model is fitted on the target-train split using their true class labels, its accuracy is used as a gold standard.  As expected, the source model is affected significantly by the spurious correlation present in the source dataset. The reweight oracle effectively improves on the test accuracy, but it requires knowledge on the $A$ and the true importance weights. Our IW and OR models were also able to effectively improve on the test accuracy and get a performance comparable to both the oracles, target model and reweighting models, without requiring the $A$, demonstrating the effectiveness of our exponential tilt framework. The DR model, on the other hand, fails to do so. 
As a possible reason, we speculate that a problem may arise within the following part of the DR logistic regression loss optimization in \eqref{eq:logistic-IW-OR-DR}:
\[
\begin{aligned}
    & \textstyle \frac{1}{n_1}\sum_{i = 1}^n R_i \widehat \omega(X_i, Y_i) \{\ell(X_i, Y_i; \beta)  - \ell(X_i, \widehat m_0(X_i); \beta) \}\\
    & \textstyle = - \frac{1}{n_1}\sum_{i = 1}^n R_i \widehat \omega(X_i, Y_i) \{Y_i - \widehat m_0(X_i)\} \{\beta_0 + \beta_1^\top X_i\}
\end{aligned}
\] The $Y_i - \widehat m_0(X_i)$ is negative whenever $Y_i = 0$ and $\widehat m_0(X_i) > 0$. If $\beta_1^\top X_i > 0$ for such an observation, then the risk is minimized at $\|\beta_1\|_2 \to \infty$ making the minimization unstable. Neither IW nor OR estimates suffer from this problem, so they were able to get a performance comparable to the gold standard.

\paragraph{Did we estimate the `missingness controlling variable (MCV)' efficiently?} Recall that the distribution of Waterbirds satisfies subpopulation shift with $X \perp R \mid (Y, A)$, where the $A$ is referred to as the MCV. Since under this assumption the true importance weights depend on $(Y, A)$ and we did not use $A$ in our importance weights, we explore whether our exponential tilt model produces a reasonable surrogate for $A$. Our estimated importance weights are $\widehat \omega(X, Y) = \exp(\widehat \alpha_Y + \widehat \beta_Y ^\top X)$, which depend on $X$ only through the $\widehat U = (\widehat \beta_0 ^\top X, \widehat \beta_1 ^\top X)$. We want to check whether,  among all linear transformations of $X$, $\widehat U$ is a good surrogate for $A$. At best, we can predict $A$ using a logistic regression model with $X$. To check whether $\widehat U$ is a good surrogate, 
we compare the following: (1) accuracy of the penalized logistic regression model for $A \sim X$, \ie\ for predicting $A$ using $X$, and (2) the same for the logistic regression model $A \sim \widehat U$. Both the models are trained on the target-train split and then their accuracies are evaluated on the target-test split.  The $A \sim \widehat U $ model has the test accuracy of $93.67\pm 0.50 \%$, which is close compared to the accuracy $94.48\pm 0.50\%$ for the $A \sim X$ model. Thus, we conclude that the exponential tilt parameters are able to approximate the MCV with reasonable accuracy.

%% file: sections/discussion.tex
\section{Discussion}
\label{sec:discussion}

We study the problem of missing not at random dataset with binary outcomes. Our approach utilizes an exponential tilt method, which eliminates the need for prior knowledge about missingness instruments or shadow variables. We introduce a sufficient condition that allows the tilt parameters to be identified and describe an algorithm to estimate these parameters. We propose importance weighted and double robust estimates for certain mean functions of interest and validate their performances in a synthetic dataset. In an experiment involving the Waterbirds dataset, we utilize our exponential tilt framework to mitigate the problem of spurious correlation and improve prediction performance in the target domain. Moreover, the same experiment illustrates that the exponential tilt model can effectively approximate the variable controlling for missingness.

The problem of missing data is intricately connected to several domains within statistical science, including causal inference, unsupervised transfer learning, survey sampling, and more. While this paper explores its relationship with unsupervised transfer learning, in a future investigation, it would be interesting to see the implications of the exponential tilt model in other areas. 


%% file: sections/emprical-likelihood.tex
\section{Empirical likelihood approach}
\label{sec:empirical-likelihood}

We use the empirical likelihood method to estimate the tilt parameters. Let us define $p_i = P(X= X_i \mid R = 1)$ as the marginal probability mass at $X_i$ and $\eta_i = P(Y = 1\mid R = 1, X = X_i)$. Then according to the exponential tilt 
\[
\textstyle P(X_i \mid R = 0) = p_i \eta_i \exp(\alpha_1 + \beta_1^\top T_i) + p_i (1 - \eta_i) \exp(\alpha_0 + \beta_0 ^\top T_i) = p_i \kappa(\eta_i, T_i; \theta) 
\] where $\kappa(\eta_i, T_i; \theta) = \eta_i \exp(\alpha_1 + \beta_1^\top T_i) +  (1 - \eta_i) \exp(\alpha_0 + \beta_0 ^\top T_i)$ With $\pi_r = P(R = 1)$ the empirical likelihood is 
\begin{equation}\label{eq:empicical-likelihood}
\begin{aligned}
  L(\theta, \bp, \pi_r)  & = \textstyle \prod_{i = 1}^n \pi_r ^{R_i} (1 - \pi_r)^{1 - R_i} \times \textstyle \prod_{i = 1}^n  \big \{   (p_i \eta_i) ^{Y_i} (p_i (1 - \eta_i)) ^{1 - Y_i} \big\}^{R_i} \\
  & \textstyle \quad \times \prod_{i = 1}^n \textstyle \big\{ p_i \kappa(\eta_i, T_i; \theta) \big\}^{1- R_i}
\end{aligned}
\end{equation} 
For the moment, assume that we have estimates $\widehat \eta_i = \eta_1(X_i; \widehat \xi)$ or $\eta_i = P(Y = 1\mid R = 1, X = X_i)$ that are fitted within a model family $\{\eta_1(x; \xi): \xi \in \Xi\}$. With that, our empirical likelihood estimation problem is 
\begin{equation} \label{eq:empirical-likelihood}
    \begin{aligned}
      & \textstyle   \max_{\theta, \bp, \pi_r} ~~ \ell(\theta, \bp, \pi_r) \{ = \log L(\theta, \bp, \pi_r)\} \\
      & \textstyle \text{subject to} ~~ \sum_{i =1}^n p_i = 1, ~~ \sum_{i = 1}^n p_i \kappa(\widehat \eta_i, T_i; \theta) = 1\,,
    \end{aligned}
\end{equation} where $\bp = (p_1, \dots, p_n)$. The following lemma simplifies the likelihood maximization in terms of a profile likelihood. 

\begin{lemma}\label{lemma:profile-likelihood} After optimizing over $\bp$ and $\pi_r$ the profile likelihood of \eqref{eq:empicical-likelihood} for optimizing $\theta$ is 
\begin{equation} \label{eq:profile-likelihood}
 \textstyle    \ell_{\mathrm{prof}}(\theta) = \sum_{i = 1}^n \big[- \log\{n_1 + n_0 \kappa(\widehat \eta_i, T_i; \theta)\} + (1 - R_i)  \log \kappa(\widehat \eta_i, T_i; \theta)\big]
\end{equation}
    
\end{lemma}

We consider synthetic experiments similar to Section \ref{sec:simulation} with the well-specified (eq. \eqref{eq:sim-well-sp}) and misspecified (eq. \ref{eq:sim-missp}) models. The results are presented in Figure \ref{fig:well-sp-simulation-el} for well-specified model and in Figure \ref{fig:missp-simulation-el} for misspecified model. Comparing them to the Figures \ref{fig:well-sp-simulation} and \ref{fig:missp-simulation}, the estimates are generally worse when we use profile likelihood in eq. \eqref{eq:profile-likelihood} to estimate the tilt parameters, as opposed to estimating them using Algorithm \ref{alg:exp-grad}.

\begin{figure}[h!]
    \centering
    \includegraphics[width=1\linewidth]{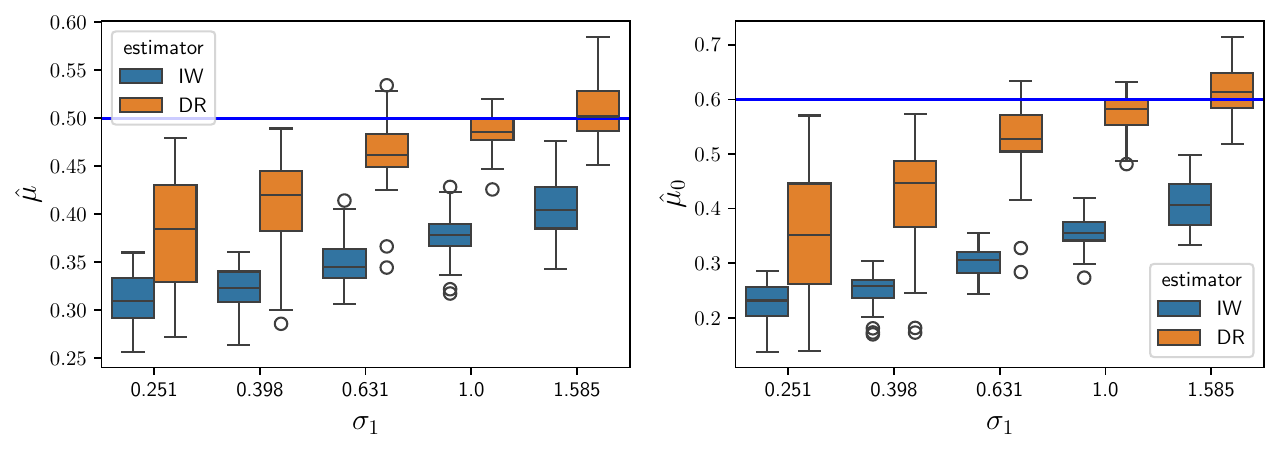}
    \caption{Different estimates in well-specified exponential tilt model in \eqref{eq:sim-well-sp} using tilt parameter estimates obtained from the empirical likelihood approach in \eqref{eq:profile-likelihood}.}
    \label{fig:well-sp-simulation-el}
\end{figure}

\begin{figure}[h!]
    \centering
    \includegraphics[width=1\linewidth]{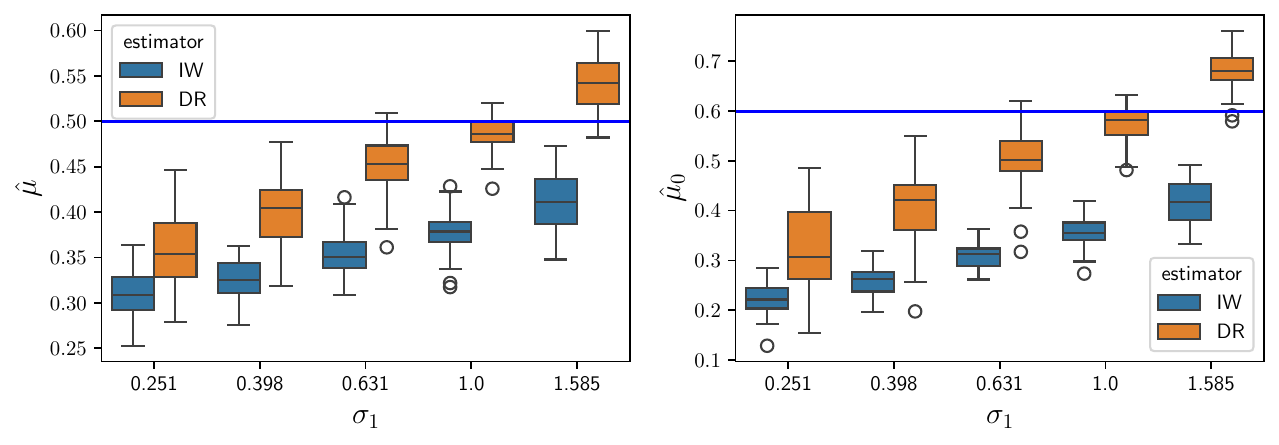}
    \caption{Different estimates in misspecified exponential tilt model in \eqref{eq:sim-missp} using tilt parameter estimates obtained from the empirical likelihood approach in \eqref{eq:profile-likelihood}.}
    \label{fig:missp-simulation-el}
\end{figure}

\begin{proof}[Proof of Lemma \ref{lemma:profile-likelihood}]
    The optimal value of $\pi_r$ within \eqref{eq:empicical-likelihood} is $\widehat \pi_r = \nicefrac{n_1}{n}$. To optimize over $p_i$'s, let's look at the Lagrangian, which is 
    \[
    \begin{aligned}
    \textstyle  &   \textstyle  \sum_{i = 1}^n \log p_i + \sum_{i = 1}^n (1 - R_i) \log \kappa(\widehat \eta_i, T_i; \theta)  - \lambda_1 \big(\sum_{i = 1}^n p_i - 1\big) - \lambda_2 \big(\sum_{i = 1}^n p_i \kappa(\widehat \eta_i, T_i; \theta) - 1\big)\\
    & \textstyle =   \sum_{i = 1}^n \log p_i  + \sum_{i = 1}^n (1 - R_i) \log \kappa(\widehat \eta_i, T_i; \theta)- \lambda_1 \big(\sum_{i = 1}^n p_i - 1\big) - \lambda_2 \sum_{i = 1}^n p_i \{\kappa(\widehat \eta_i, T_i; \theta) - 1\} 
    \end{aligned}
    \] Setting the partial derivatives of the Lagrangian with respect to $p_i$, $\lambda_1$ and $\lambda_2$ to zero, we obtain 
    \[
    \begin{aligned}
        0 & = \textstyle \frac1{p_i} - \lambda_1 - \lambda_2 \{\kappa(\widehat \eta_i, T_i; \theta) - 1\}\\
        0 & = \textstyle \sum_{i = 1}^n p_i - 1\\
        0 & = \textstyle \sum_{i = 1}^n p_i \kappa(\widehat \eta_i, T_i; \theta) - 1\,.
    \end{aligned}
    \]
    Thus $p_i = [\lambda_1 + \lambda_2 \{\kappa(\widehat \eta_i, T_i; \theta) - 1\}]^{-1} $. 
    From the first equation, we obtain the following:
    \[
    \begin{aligned}
        0 & =\textstyle  \sum_{i = 1}^n p_i \times \big(\frac1{p_i} - \lambda_1 - \lambda_2 \{\kappa(\widehat \eta_i, T_i; \theta) - 1\}\big) \\
        & = \textstyle n - \lambda_1 - \sum_{i = 1}^n \lambda_2 p_i\{\kappa(\widehat \eta_i, T_i; \theta) - 1\} = n - \lambda_1
    \end{aligned}
    \] because $\sum_{i = 1}^n  p_i\{\kappa(\widehat \eta_i, T_i; \theta) - 1\} = 0$. This implies $\lambda_1 = n$. To obtain $\lambda_2$ we differentiate the Lagrangian with respect to $\alpha_0$ and $\alpha_1$. At the optimum, they should be zero, leading to the equations 
    \[
    \begin{aligned}
       0 & = \textstyle  \sum _{i = 1}^n (1 - R_i) \frac{\exp(\alpha_0 + \beta_0 ^\top T_i)(1 - \widehat \eta_i)}{\kappa(\widehat \eta_i, T_i; \theta)} - \lambda_2 \sum_{i = 1}^n p_i \exp(\alpha_0 + \beta_0 ^\top T_i)(1 - \widehat \eta_i)\\
       0 & = \textstyle  \sum _{i = 1}^n (1 - R_i) \frac{\exp(\alpha_1 + \beta_1 ^\top T_i) \widehat \eta_i}{\kappa(\widehat \eta_i, T_i; \theta)} - \lambda_2 \sum_{i = 1}^n p_i \exp(\alpha_1 + \beta_1 ^\top T_i) \widehat \eta_i
    \end{aligned}
    \] By adding them together we obtain
    \[
    \begin{aligned}
         \textstyle 0 & = \textstyle  \sum _{i = 1}^n (1 - R_i) \frac{\exp(\alpha_0 + \beta_0 ^\top T_i)(1 - \widehat \eta_i) + \exp(\alpha_1 + \beta_1 ^\top T_i) \widehat \eta_i}{\kappa(\widehat \eta_i, T_i; \theta)} - \lambda_2 \sum_{i = 1}^n p_i \kappa(\widehat \eta_i, T_i; \theta)\\
         & = \textstyle \sum _{i = 1}^n (1 - R_i)  - \lambda_2 = n_0 - \lambda_2,
    \end{aligned}
    \] leading to $\lambda_2 = n_0$. We now plug in the values of $\lambda_1 = n$ and $\lambda_2 = n_0$ to obtain 
    \[
    \textstyle p_i = \frac{1}{ n + n_0 \{\kappa(\widehat \eta_i, T_i; \theta) - 1\}} = \frac{1}{ n_1  + n_0 \kappa(\widehat \eta_i, T_i; \theta)}
    \] Finally, we plugin these values of $p_i$ and obtain 
    \[
    \begin{aligned}
        & \textstyle \sum_{i = 1}^n \log p_i + \sum_{i = 1}^n (1 - R_i) \log \kappa(\widehat \eta_i, T_i; \theta)\\
        & = \textstyle \sum_{i = 1}^n - \log \{ n_1  + n_0 \kappa(\widehat \eta_i, T_i; \theta)\} + \sum_{i = 1}^n (1 - R_i) \log \kappa(\widehat \eta_i, T_i; \theta)
    \end{aligned}
    \] This completes the proof. 
\end{proof}

%% file: sections/exp-supp.tex
\section{Further implementation details of the experiments}

\label{sec:exp-supp}
The codes can be found in the repository \url{https://github.com/smaityumich/MNAR_binary_eponential_tilt}.
The Algorithm \ref{alg:exp-grad} is implemented within the \texttt{exp\_grad.py} file, and has the following structure:


\begin{verbatim}
def exponentiated_gradient(str_x, str_y, str_p, ttr_x, ttr_p, 
                             eps = 1e-3, tol = 2e-3, B = 5, 
                             lr = 4e-3, max_iter = 4000, 
                             verbose = 0, reg = 1e-5):
    """
    str_x, str_y, str_p: x, y, proababilites from source dataset
    ttr_x, ttr_p: x and probabilities from target-train dataset
    eps: constraint relaxation
    tol: tolerance levels for convergence
    B: maximum value for Lagrange multiplier
    lr: learning rate
    reg: regularization strength
    """
\end{verbatim}

The \texttt{reg} parameter allows an option to add a regularization term $(\nicefrac{\lambda}2)\|\theta\|$ to the optimization in eq. \eqref{eq:optimization-v1}, which changes the gradient update of $\theta$ in Algorithm \ref{alg:exp-grad}, line 6 as 
\[
\textstyle \theta^{(t+1)} \gets (1 - \rho_1 \lambda)\theta^{(t)} - \rho_1 \partial_\theta \Lambda(\theta^{(t)})
\] The $\lambda$ can be set to zero for no regularization. The parameter values for the experiments are provided within the code implementations. 

%% file: sections/tech_results.tex
\section{Technical lemmas and proofs}
\label{sec:tech_results}

\subsection{Technical lemmas}

\begin{lemma}\label{lemma:prop-score}
    Under the exponential tilt model \eqref{eq:exponential-tilt} the missingness propensity score is $P(R = 0\mid X = x, Y = y)  = \frac{{P(R = 0)} \exp(\alpha_y^\star+ {\beta_y^\star} ^\top x )}{ {P(R = 1)}+ {P(R = 0)} \exp(\alpha_y^\star+ {\beta_y^\star} ^\top x )}$. 
\end{lemma}

\begin{proof}[Proof of Lemma \ref{lemma:prop-score}]
    Since  \[
    \begin{aligned}
        \textstyle \frac{P(R = 0\mid X = x, Y = y)}{P(R = 1\mid X = x, Y = y)} = \frac{P(R = 0)dP(x\mid Y = y, R = 0) P(Y = y\mid R = 0)}{P(R = 1)dP(x\mid Y = y, R = 1) P(Y = y\mid R = 1)} = \frac{P(R = 0)}{P(R = 1)} \exp(\alpha_y^\star+ {\beta_y^\star} ^\top x )
    \end{aligned}
    \] we have 
    \[
    \textstyle P(R = 0\mid X = x, Y = y) = \frac{\frac{P(R = 0)}{P(R = 1)} \exp(\alpha_y^\star+ {\beta_y^\star} ^\top x )}{ 1+ \frac{P(R = 0)}{P(R = 1)} \exp(\alpha_y^\star+ {\beta_y^\star} ^\top x )} = \frac{{P(R = 0)} \exp(\alpha_y^\star+ {\beta_y^\star} ^\top x )}{ {P(R = 1)}+ {P(R = 0)} \exp(\alpha_y^\star+ {\beta_y^\star} ^\top x )}
    \]
\end{proof}

\begin{lemma}\label{lemma:shadow-variable} For $B = (\beta_0^\star, \beta_1 ^\star) \in \reals^{d \times 2 }$, $U = B^\top X$ and $V = (\bI - \bP_{B})X$ where $\bP_B$ is an orthogonal projection on the column space of $B$, the exponential tilt assumption \eqref{eq:exponential-tilt} implies  
\[
\textstyle V \perp R \mid Y, U ~~\text{and hence}~~ X\perp R \mid Y, U\,.
\]
    
\end{lemma}

\begin{proof}[Proof of Lemma \ref{lemma:shadow-variable}]
    Implementing a change of variable density $p(x \mid Y = y, R = r)$ can be expressed as 
    \[
    \textstyle p(x \mid Y = y, R = r) = p' (u, v \mid Y = y, R = r) J
    \] where $J$ is the Jacobian that is independent on $u, v$. Then the expoential tilt model \eqref{eq:exponential-tilt} implies 
    \[
    \textstyle \frac{p'(u, v \mid Y = y, R = 0)}{p'(u, v \mid Y =y, R = r )} = \exp(\alpha_y^\star + u_y)
    \] where $U = (U_0, U_1)$ and $u = B \top x = (u_0, u_1)$. This implies
    \[
    \begin{aligned}
        \textstyle p'(u \mid Y = y, R = 0) & \textstyle = \int p'(u, v \mid Y = y, R = 0)du\\
        & \textstyle = \exp(\alpha_y^\star + u_y)\int p'(u, v \mid Y = y, R = 1)dv\\
        & = \exp(\alpha_y^\star + u_y)p'( u \mid Y = y, R = 1)\\
    \end{aligned}
    \] and finally we have 
    \[
    \begin{aligned}
       p'(v \mid U = u,  Y = y, R = 0) & =  \textstyle \frac{p'(u, v \mid Y = y, R = 0)}{p'(u \mid Y = y, R = 0)}\\
       & = \textstyle \frac{\exp(\alpha_y^\star + u_y)p'(u, v \mid Y = y, R = 1)}{\exp(\alpha_y^\star + u_y)p'(u \mid Y = y, R = 1)}\\
       & = p'(v \mid U = u,  Y = y, R = 1)
    \end{aligned}
    \] which implies $V \perp R \mid U, Y$. Since, $p'(u, v\mid R = r, Y = y, U = u) = p'(v \mid R = r, Y = y, U = u)$, we have $(U, V) \perp R \mid Y, U$. 
\end{proof}

\begin{lemma}
    \label{lemma:tech2}
Under the exponential tilt model \eqref{eq:exponential-tilt}
    \[
    \textstyle \text{logit} \{ P(Y = 1 \mid X, R = 0) \} - \text{logit} \{ P(Y = 1 \mid X, R = 1) \} = \log \omega(X, 1) - \log \omega(X, 0)
    \]
\end{lemma}

\begin{proof}[Proof of Lemma \ref{lemma:tech2}]
    Let's start from the following equality: 
    \[
    \begin{aligned}
        & \textstyle \text{logit} \{ P(Y = 1 \mid X, R = 0) \} - \text{logit} \{ P(Y = 1 \mid X, R = 1) \} \\
        & = \textstyle \log \big\{ \frac{P(Y = 1 \mid X, R = 0)P(Y = 0 \mid X, R = 1)}{P(Y = 0 \mid X, R = 0)P(Y = 1 \mid X, R = 1)}\big\} \\
        & = \textstyle \textstyle \log \big\{ \frac{p(X, Y = 1 \mid  R = 0)p (X, Y = 0 \mid  R = 1)}{p(X, Y = 0 \mid  R = 0)P(X, Y = 1 \mid  R = 1)}\big\}\\
        & = \textstyle \textstyle \log \big\{ \frac{p(X, Y = 1 \mid  R = 0)}{P(X, Y = 1 \mid  R = 1)}\times \frac{p (X, Y = 0 \mid  R = 1)}{p(X, Y = 0 \mid  R = 0)}\big\} = \textstyle \log \big\{\frac{\omega(X, 1)}{\omega(X, 0)}\big\}\,.
    \end{aligned}
    \] This completes the proof. 
\end{proof}

\begin{lemma}\label{lemma:uniqueness}
    There is a unique solution of $\theta$ to the equation \eqref{eq:marginal}
    \[
    \textstyle \textstyle p(x\mid R = 0) = \sum_{y \in\{0, 1\}} \exp(\alpha_y + {\beta_y}^\top t) p(x, Y = y\mid R = 1)\,
    \]
if the   assumption \ref{assump:identifiability} is satisfied.  
\end{lemma}
\begin{proof}[Proof of Lemma \ref{lemma:uniqueness}] We shall prove the lemma by contradiction. Let's assume that there are two possible values  $\theta = (\alpha_0, \alpha_1, \beta_0, \beta_1)$ and $\theta' = (\alpha_0', \alpha_1', \beta_0', \beta_1')$ that satisfy the equation. This leads to the equality: 
\begin{equation}\label{eq:two-theta}
    \begin{aligned}
    & p(x, Y = 1\mid R = 1) \big \{ \exp(\alpha_1 + {\beta_1}^\top t) - \exp(\alpha_1' + {\beta_1'}^\top t) \big\} \\
    & = p(x, Y = 0\mid R = 1) \big \{ \exp(\alpha_0' + {\beta_0'}^\top t) - \exp(\alpha_0 + {\beta_0}^\top t) \big\}\\ 
\end{aligned}
\end{equation}
or equivalently, 
\begin{equation}\label{eq:tech-lemma3}
    \begin{aligned}
     \textstyle \frac{p(x, Y = 1\mid R = 1)}{p(x, Y = 0\mid R = 1)} = & \textstyle \frac{\exp(\alpha_0' + {\beta_0'}^\top t) - \exp(\alpha_0 + {\beta_0}^\top t)}{\exp(\alpha_1 + {\beta_1}^\top t) - \exp(\alpha_1' + {\beta_1'}^\top t)}\\
     \textstyle \text{or,} ~\log \big \{ \frac{p(x, Y = 1\mid R = 1)}{p(x, Y = 0\mid R = 1)}\big\} & = (\alpha_0 - \alpha_1')  + (\beta_0 - \beta_1')^\top t + \log \big\{\textstyle  \frac{\exp(\alpha_0' + {\beta_0'}^\top t - \alpha_0 - {\beta_0}^\top t) - 1}{\exp(\alpha_1 + {\beta_1}^\top t - \alpha_1' - {\beta_1'}^\top t) - 1}\big\} 
\end{aligned}
\end{equation}

which violates the condition in Assumption \ref{assump:identifiability}. This proves  the lemma. 



\end{proof}

\begin{lemma}
    \label{lemma:KL-simplification}
    Let us denote the KL loss 
    \[
    \textstyle \text{KL} \big(p(x\mid R = 0); \exp(\alpha_1 + \beta_1 ^\top t) p(x ,  y = 1\mid  R = 1) + \exp(\alpha_0 + \beta_0 ^\top t) p(x ,  y = 0\mid  R = 1)\big)
    \] as $\bR(\theta)$. The following minimization is equivalent to the optimization in \eqref{eq:optimization}
    \[
    \begin{aligned}
        & \min_\theta \bR(\theta)\\
        & \text{subject to} ~~ \textstyle \int \{\exp(\alpha_1 + \beta_1 ^\top t) p(x ,  y = 1\mid  R = 1) + \exp(\alpha_0 + \beta_0 ^\top t) p(x ,  y = 0\mid  R = 1)\}  dx = 1\,.
    \end{aligned}
    \] 
\end{lemma}

\begin{proof}[Proof of Lemma \ref{lemma:KL-simplification}]
    Let's start with the constraint. Defining $\eta_1(x) = P(Y = 1\mid X = x, R = 1)$ it simplifies to
    \[
    \begin{aligned}
        & \textstyle \int \{\exp(\alpha_1 + \beta_1 ^\top t) p(x ,  y = 1\mid  R = 1) + \exp(\alpha_0 + \beta_0 ^\top t) p(x ,  y = 0\mid  R = 1)\}  dx\\
        & = \textstyle \int p(x\mid R = 1) \{\exp(\alpha_1 + \beta_1 ^\top t) \eta_1(x) + \exp(\alpha_0 + \beta_0 ^\top t) (1 - \eta_1(x))\}  dx\\
        & = \textstyle \Ex[\exp(\alpha_1 + \beta_1 ^\top T) \eta_1(X) + \exp(\alpha_0 + \beta_0 ^\top T) (1 - \eta_1(X)) \mid R=1]\,,
    \end{aligned}
    \] which is identical to the constraint in \eqref{eq:optimization}. Now, let's focus on the KL divergence. 
    \[
    \begin{aligned}
        & \textstyle \text{KL} \big(p(x\mid R = 0); \exp(\alpha_1 + \beta_1 ^\top t) p(x ,  y = 1\mid  R = 1) + \exp(\alpha_0 + \beta_0 ^\top t) p(x ,  y = 0\mid  R = 1)\big)\\
        & \textstyle = \int \log \{p(x\mid R = 0)\} p(x\mid R = 0) dx \\
        & \textstyle \quad - \int p(x\mid R = 0) \log \big \{\exp(\alpha_1 + \beta_1 ^\top t) p(x ,  y = 1\mid  R = 1) \\
        & \quad \quad + \exp(\alpha_0 + \beta_0 ^\top t) p(x ,  y = 0\mid  R = 1)\big\} dx\\
        & \textstyle = \int \log \{p(x\mid R = 0)\} p(x\mid R = 0) dx - \int \log \{p(x\mid R = 1)\} p(x\mid R = 0) dx \\
        & \textstyle \quad - \int p(x\mid R = 0) \log \big \{\exp(\alpha_1 + \beta_1 ^\top t)  \eta_1(x) + \exp(\alpha_0 + \beta_0 ^\top t) (1 - \eta_1(x))\big\} dx
    \end{aligned}
    \] The $\int \log \{p(x\mid R = 0)\} p(x\mid R = 0) dx - \int \log \{p(x\mid R = 1)\} p(x\mid R = 0) dx$ doesn't involve any parameters and thus can be ignored. The remaining part is essentially 
    \[
    \textstyle \Ex \big [- \log \big \{\exp(\alpha_1 + \beta_1 ^\top T)  \eta_1(X) + \exp(\alpha_0 + \beta_0 ^\top T) (1 - \eta_1(X))\big\}  \bigm \vert R = 0 \big ]\,,
    \] which is identical to the objective in \eqref{eq:optimization}. This completes the proof. 
\end{proof}

\begin{lemma} \label{lemma:tech1}
For any function $\nu(x, y)$ it holds
    \begin{equation}
       \textstyle  \frac{\Ex[R \exp(\gamma(Y \mid X)) \nu(X, Y)\mid X = x ]}{\Ex[R \exp(\gamma(Y \mid X)) \mid X = x]} = \Ex[\nu(X, Y) \mid X = x, R = 0]
    \end{equation}
\end{lemma}

\begin{proof}[Proof of Lemma \ref{lemma:tech1}]
 Notice that, 
\[
\begin{aligned}
    & \Ex[R \exp(\gamma(Y \mid X)) \nu(X, Y)\mid X = x]\\
    & = P(R = 1\mid X = x) \Ex[ \exp(\gamma(Y \mid X)) \nu(X, Y) \mid X = x, R = 1]\\
    & = \textstyle  P(R = 1\mid X  = x)\sum_{y \in \{0, 1\}} e^{\gamma(y \mid x)} \nu(x, y) P(Y = y\mid X = x, R = 1)\\
    & = \textstyle  \eta_r(x)\sum_{y \in \{0, 1\}}  \nu(x, y) P(Y = y\mid X = x, R = 1) \frac{P(R = 0\mid Y = y, X = x) P(R = 1\mid Y = 0, X = x)}{P(R = 1\mid Y = y, X = x) P(R = 0\mid Y = 0, X = x)}\\
    & = \textstyle  \eta_r(x)\sum_{y \in \{0, 1\}}  \nu(x, y) P(Y = y\mid X = x, R = 1) \frac{P(Y = y\mid R = 0, X = x) P(Y = 0\mid R = 0, X = x)}{P(Y = y \mid R = 1, X = x) P(Y = 0 \mid R = 0, X = x)}\\
    & = \textstyle  \eta_r(x) \frac{P(Y = 0\mid R = 0, X = x)}{P(Y = 0 \mid R = 0, X = x)}\sum_{y \in \{0, 1\}}  \nu(x, y) P(Y = y\mid R = 0, X = x) \\
    & = \textstyle  \eta_r(x) \frac{P(Y = 0\mid R = 0, X = x)}{P(Y = 0 \mid R = 0, X = x)}\Ex[  \nu(X, Y) \mid X = x, R = 0]
\end{aligned}
\] and thus 
\[
\begin{aligned}
    & \textstyle \frac{\Ex[R \exp(\gamma(Y \mid X)) \nu(X, Y)\mid X = x]}{\Ex[R \exp(\gamma(Y \mid X)) \mid X = x]}\\
    & = \textstyle \frac{\eta_r(x) \frac{P(Y = 0\mid R = 0, X = x)}{P(Y = 0 \mid R = 0, X = x)}\Ex[  \nu(X, Y) \mid X = x, R = 0]}{\eta_r(x) \frac{P(Y = 0\mid R = 0, X = x)}{P(Y = 0 \mid R = 0, X = x)}\Ex[  1 \mid X = x, R = 0]} = \Ex[  \nu(X, Y) \mid X = x, R = 0]
\end{aligned}
\]
\end{proof}

\begin{lemma}\label{lemma:tech3}
If $\gamma(1 \mid x) = \gamma(1 \mid x, \theta^\star)$ then for any $\nu(X, Y)$ it holds: 
\[
\textstyle \frac{\Ex \big[R \big(\frac{1}{\eta_r(X, Y; \theta^\star)} - 1\big) \nu(X, Y) \big\vert X\big] }{\Ex \big[R \big(\frac{1}{\eta_r(X, Y; \theta^\star)} - 1\big)  \big\vert X\big] } = \Ex [\nu(X, Y) \mid X, R = 0]\,. 
\]
  
\end{lemma}

\begin{proof}[Proof of Lemma \ref{lemma:tech3}]
    Since    $\gamma(1 \mid x)$ is well specified, we have 
    \[
    \begin{aligned}
       \textstyle  \frac{1}{\eta_r(X, Y; \theta^\star)} - 1 = \frac{1 - \eta_r(X, Y; \theta^\star)}{\eta_r(X, Y; \theta^\star)} = \exp(\gamma(Y \mid X)) \frac{1 - \eta_{r}(X, 0; \theta^\star)}{\eta_{r}(X, 0; \theta^\star)} = \exp(\gamma(Y \mid X)) \kappa(X)
    \end{aligned}
    \] where $\kappa (X) = \{1 - \eta_{r}(X, 0; \theta^\star)\}/{\eta_{r}(X, 0; \theta^\star)}$. Thus, we have the equality Lemma \ref{lemma:tech1} to notice that 
    \[
    \begin{aligned}
        \textstyle \frac{\Ex \big[R \big(\frac{1}{\eta_r(X, Y; \theta^\star)} - 1\big) \nu(X, Y) \big\vert X\big] }{\Ex \big[R \big(\frac{1}{\eta_r(X, Y; \theta^\star)} - 1\big)  \big\vert X\big] }& = \textstyle \frac{\Ex \big[R \exp(\gamma(Y \mid X)) \kappa(X) \nu(X, Y) \big\vert X\big] }{\Ex \big[R \exp(\gamma(Y \mid X)) \kappa(X)  \big\vert X\big] }\\
        & = \textstyle \frac{\Ex \big[R \exp(\gamma(Y \mid X))  \nu(X, Y) \big\vert X\big] }{\Ex \big[R \exp(\gamma(Y \mid X))   \big\vert X\big] }\\
        & = \textstyle \Ex [\nu(X, Y) \mid X, R = 0]
    \end{aligned}
    \] where the last equlity is justified by Lemma \ref{lemma:tech1}. 
\end{proof}

%% file: sections/tech_results_2.tex
\subsection{Proof of lemmas and theorems in Section \ref{sec:estimation}}

\subsubsection{Proof of Lemma \ref{lemma:propensity-score}}

\begin{proof}[Proof of Lemma \ref{lemma:propensity-score}]
    Let's begin with the equality
\begin{equation}
\begin{aligned}
& \textstyle   \frac{P(R = 0\mid X =x , Y = y)}{P(R = 1\mid X =x , Y = y)} = \textstyle \frac{p(x, y\mid R = 0) P(R = 0)}{p(x, y\mid R = 1) P(R = 1)} = \omega(x, y) \frac{P(R = 0)}{P(R = 1)}\\
& \text{or, } \textstyle P(R = 1\mid X = x, Y = y) =  \textstyle \frac{P(R = 1)}{P(R = 1) + \omega(x, y) P(R = 0)}\,,
\end{aligned}
\end{equation} which follows, 
\[
    \textstyle \frac{R }{ P(R = 1\mid X , Y)}  \textstyle = R  \big\{1 + \omega(X, Y) \frac{P(R = 0)}{P(R = 1)}\big\}\,.
\] 
\end{proof}

\subsubsection{Proof of Lemma \ref{lemma:conditional-mean-r0}}

\begin{proof}[Proof of Lemma \ref{lemma:conditional-mean-r0}]
    The following equalities prove the lemma: 
    \begin{equation}
    \begin{aligned}
        & \textstyle \frac{P(Y = 1\mid R = 0, X = x)}{P(Y = 0\mid R = 0, X = x)} \frac{P(Y = 0\mid R = 1, X = x)}{P(Y = 1\mid R = 1, X = x)} = \frac{P(R = 0\mid Y = 1, X = x)}{P(R = 0\mid Y = 0, X = x)} \frac{P(R = 1 \mid Y = 0 , X = x)}{P(R = 1\mid Y = 1, X = x)} = \exp(\gamma(1 \mid x))\\
        & \text{or,} ~~ \textstyle  \frac{P(Y = 1\mid R = 0, X = x)}{P(Y = 0\mid R = 0, X = x)} = \exp(\gamma(1 \mid x))  \frac{P(Y = 1\mid R = 1, X = x)}{P(Y = 0\mid R = 1, X = x)} \\
        & \text{or,} ~~ \textstyle P(Y = 1\mid R = 0, X = x) = \frac{\exp(\gamma(1 \mid x))  {P(Y = 1\mid R = 1, X = x)}}{\exp(\gamma(1 \mid x))  {P(Y = 1\mid R = 1, X = x)} + P(Y = 0\mid R = 1, X = x)  }\,.
    \end{aligned}
\end{equation}
\end{proof}

\subsubsection{Proof of Lemma \ref{lemma:DR-bias}}
\begin{proof}[Proof of Lemma \ref{lemma:DR-bias}]
Both the estimators are sample averages with respect to appropriate scores, as
\begin{equation}\label{eq:score-rep}
    \begin{aligned}
        & \textstyle \widehat \mu_{0, \psi} = \frac{1}{n} \sum_{i = 1}^n\delta_{0, \psi}(X_i, Y_i, R_i; \widehat \theta, \widehat \xi), ~~~ \widehat \mu_{\psi} = \frac{1}{n} \sum_{i = 1}^n\delta_{\psi}(X_i, Y_i, R_i; \widehat \theta, \widehat \xi), ~~ \psi \in \{ \text{IW}, \text{DR}\} \\
&  \text{where,} \textstyle ~~    \delta_{0, \text{IW}}(X_i, Y_i, R_i; \widehat \theta, \widehat \xi) = \frac{R_i}{\widehat \pi_r} \omega(X_i, Y_i; \widehat \theta, \widehat \xi) \tau(X_i, Y_i)   , \\
& \textstyle \delta_{ \text{IW}}(X_i, Y_i, R_i; \widehat \theta, \widehat \xi) = {R_i \big \{ 1 + \frac{1- \widehat \pi_r}{\widehat \pi_r}  \omega(X_i, Y_i; \widehat \theta) \big \} \tau(X_i, Y_i)} \\
& \textstyle\delta_{0, \text{DR}}(X_i, Y_i, R_i; \widehat \theta, \widehat \xi) = \frac{R_i}{\widehat \pi_r} \omega(X_i, Y_i; \widehat \theta) \{\tau(X_i, Y_i) -  m_{0, \tau} (X_i; \widehat \theta, \widehat \xi)\} +  \frac{1 - R_i}{1 - \widehat \pi_r}  m_{0, \tau}(X_i; \widehat \theta, \widehat \xi)\\
& \textstyle\delta_{\text{DR}}(X_i, Y_i, R_i; \widehat \theta, \widehat \xi) = R_i \{\frac{1 - \widehat \pi_r}{\widehat \pi_r} + \widehat \omega(X_i, Y_i)\}  \{ \tau(X_i, Y_i) -  m_{0, \tau}(X_i; \widehat \theta, \widehat \xi)\} +    m_{0, \tau}(X_i; \widehat \theta, \widehat \xi)\,.
    \end{aligned}
\end{equation} 
Within the proof we shall use the equality multiple times
\[
\textstyle \Ex[\tau(X, Y)\mid R = 0] = \Ex\big[\frac{1 - R}{P(R = 0)}\tau(X, Y)\big]\,.
\]
    If $\hat \theta \to \theta^\star$ and $\hat \xi \to \xi^\star$ in probability at $n \to \infty$, then from the eq. \eqref{eq:score-rep} we have 
    \[
    \begin{aligned}
        \textstyle \Ex [\widehat \mu_{0, \text{IW}}] - \mu_0 \stackrel{p}{\longrightarrow} & \textstyle \Ex \big [\frac{R }{P(R = 1)} \omega(X, Y; \theta^\star) \tau(X, Y)  \big] - \Ex[\tau(X, Y) \mid R = 0]\\
        & \textstyle = \Ex \big [\big \{\frac{R }{P(R = 1)} \omega(X, Y; \theta^\star) - \frac{1 - R}{P(R = 0)} \big \}  \tau(X, Y)  \big]\\
        \textstyle \Ex [\widehat \mu_{ \text{IW}}] - \mu \stackrel{p}{\longrightarrow} & \textstyle \Ex \big [R \big \{1 + \frac{P(R = 0) }{P(R = 1)} \omega(X, Y; \theta^\star)\big\}  \tau(X, Y)  \big] - \Ex[\tau(X, Y) ]\\
        \textstyle & \textstyle = \Ex \big [\big \{\frac{R }{\eta_r(X, Y; \theta^\star)} - 1 \big \}  \tau(X, Y)  \big]\\
        \end{aligned}
    \]
    \[
    \begin{aligned}
        \textstyle \Ex [\widehat \mu_{0, \text{DR}}] - \mu_0 \stackrel{p}{\longrightarrow} & \textstyle \Ex \big [\frac{R }{P(R = 1)} \omega(X, Y; \theta^\star)  \big \{ \tau(X,  Y) -  m_{0, \tau}(X; \theta^\star, \xi^\star)\big \} \big]\\
        & \textstyle \quad + \Ex \big [\frac{1 - R }{P(R = 0)} \  m_{0, \tau}(X; \theta^\star, \xi^\star) \big] - \Ex[\tau(X, Y) \mid R = 0]\\
        & = \textstyle \Ex \big [\big \{\frac{R }{P(R = 1)} \omega(X, Y; \theta^\star) - \frac{1 - R}{P(R = 0)} \big \} \big \{ \tau(X,  Y) -  m_{0, \tau}(X; \theta^\star, \xi^\star)\big \} \big] \\
    \end{aligned}
    \]
    \[
    \begin{aligned}
        \textstyle \Ex [\widehat \mu_{ \text{DR}}] - \mu \stackrel{p}{\longrightarrow} & \textstyle \Ex \big [R \big \{1 + \frac{P(R = 0) }{P(R = 1)} \omega(X, Y; \theta^\star)\big\}  \big \{ \tau(X,  Y) -  m_{0, \tau}(X; \theta^\star, \xi^\star)\big \} \big]\\
        & \textstyle \quad + \Ex \big [  m_{0, \tau}(X; \theta^\star, \xi^\star) \big] - \Ex[\tau(X, Y) ]\\
        & = \textstyle \Ex \big [\big \{\frac{R }{\eta_r (X, Y, \theta^\star)}  - 1  \big \} \big \{ \tau(X,  Y) -  m_{0, \tau}(X; \theta^\star, \xi^\star)\big \} \big]
    \end{aligned}
    \]
\end{proof}

\subsubsection{Proof of Theorem \ref{thm:double-robustness}}

\begin{proof}[Proof of Theorem \ref{thm:double-robustness}] We start with {\it (i)}. If $\Ex[R\mid X, Y] = \eta_r(X, Y; \theta^\star)$ then $\frac{\Ex[R\mid X, Y]}{\eta_r(X, Y; \theta^\star)} - 1 = 0$. Thus, for any $\nu(X, Y)$ we have 
\begin{equation}
    \begin{aligned}
        \textstyle   \textstyle  \Ex\Big[\big(\frac{R}{\eta_r(X, Y; \theta^\star)} - 1 \big)\nu(X, Y)\Big]
        & = \textstyle \Ex\Big[\Ex\big[\big(\frac{R}{\eta_r(X, Y; \theta^\star)} - 1 \big)\nu(X, Y)\bigl\vert X, Y\big] \Big]\\
        & = \textstyle \Ex\Big[\big(\frac{\Ex[R\mid X, Y]}{\eta_r(X, Y; \theta^\star)} - 1 \big)\nu(X, Y) \Big] = 0\,.
    \end{aligned}
\end{equation}
We plug in $\nu(X, Y) = \tau(X, Y)$ and $\nu(X, Y) = \tau(X,  Y) -  m_{0, \tau}(X; \theta^\star, \xi^\star)$ to obtain $\text{a-bias} (\widehat \mu_{0, \text{IW}}) = 0$,  $\text{a-bias} (\widehat \mu_{ \text{IW}}) = 0$. 

The $\Ex[R\mid X, Y] = \eta_r(X, Y; \theta^\star)$ means 
\[
\textstyle \frac{P(R = 1 \mid X, Y )}{P(R = 1)} \omega(X, Y; \theta^\star) - \frac{P(R = 0\mid X, Y)}{R = 0} = 0
\] and thus for any $\nu(X, Y)$ we have 
\[
\begin{aligned}
   & \textstyle  \Ex\Big[\big(\frac{R}{P(R = 1)} \omega(X, Y; \theta^\star) - \frac{1 - R}{P(R = 0)} \big)\nu(X, Y)\Big]\\
   & = \textstyle  \Ex\Big[\Ex\Big[\big(\frac{R}{P(R = 1)} \omega(X, Y; \theta^\star) - \frac{1 - R}{P(R = 0)} \big)\nu(X, Y)\Big\vert X, Y \Big] \Big]\\
    & = \textstyle  \Ex\Big[\big(\frac{P(R = 1\mid X, Y)}{P(R = 1)} \omega(X, Y; \theta^\star) - \frac{P(R = 0\mid X, Y)}{P(R = 0)} \big)\nu(X, Y) \Big] = 0\,.
\end{aligned}
\]
We plug in $\nu(X, Y) = \tau(X, Y)$ and $\nu(X, Y) = \tau(X,  Y) -  m_{0, \tau}(X; \theta^\star, \xi^\star)$ to obtain $\text{a-bias} (\widehat \mu_{0, \text{DR}}) = 0$,  $\text{a-bias} (\widehat \mu_{ \text{DR}}) = 0$. This proves {\it (i)}. To establish {\it (ii)} we decompose the bias as

 leading to  $\text{bias}(\hat \mu) = 0$. This establishes {\it (i)}. To establish {\it (ii)} we decompose the bias as 
\begin{equation}
\begin{aligned}
     \textstyle \text{a-bias} (\widehat \mu_{ \text{DR}}) & = \underbrace{\textstyle \Ex\Big[R\big(\frac{1}{\eta_r(X, Y; \theta^\star)} - 1 \big)\big( \tau(X, Y) - m_{0, \tau}(X; \theta^\star, \xi^\star) \big)\Big]}_{\text{bias}.1}\\
     & \quad \textstyle + \underbrace{\textstyle - \Ex\big[(1 - R)\big( \tau(X, Y) - m_{0, \tau}(X; \theta^\star, \xi^\star) \big)\big]}_{\text{bias}.2}
\end{aligned}
\end{equation} The $\text{bias}.2$ is 
\begin{equation}
    \begin{aligned}
        - \text{bias}.2 & = \textstyle \Ex\big[(1 - R)\big( \tau(X, Y) - m_{0, \tau}(X; \theta^\star, \xi^\star) \big)\big]\\
        & \textstyle = \Ex\big[ \big( \tau(X, Y) - m_{0, \tau}(X; \theta^\star, \xi^\star) \big) \bigm \vert R = 0\big] P(R = 0)\\
        & \textstyle = \Ex\big[ \Ex[\tau(X, Y)\mid X, R = 0] - m_{0, \tau}(X; \theta^\star, \xi^\star) \bigm \vert R = 0\big] P(R = 0)\\
    \end{aligned}
\end{equation} which is zero because $\Ex[\tau(X, Y) \mid X, R = 0] = m_{0, \tau}(X; \theta^\star, \xi^\star)$. Now, to establish that $\text{bias}.1$ is zero we apply lemma \ref{lemma:tech3} to have the equality 
\[
\begin{aligned}
   & \textstyle \Ex\Big[R\big(\frac{1}{\eta_r(X, Y; \theta^\star)} - 1 \big)\big( \tau(X, Y) - m_{0, \tau}(X; \theta^\star, \xi^\star) \big)\Big\vert X\Big]\\
   & = \textstyle \Ex\big[ \tau(X, Y) - m_{0, \tau}(X; \theta^\star, \xi^\star) \big)\big\vert X, R = 0\big] \times \Ex \big[R \big(\frac{1}{\eta_r(X, Y; \theta^\star)} - 1\big)  \big\vert X\big] = 0 
\end{aligned}
\] because $\Ex[\tau(X, Y) \mid X, R = 0] = m_{0, \tau}(X; \theta^\star, \xi^\star)$. This establishes $\text{a-bias} (\widehat \mu_{ \text{DR}}) = 0$. To establish $\text{a-bias} (\widehat \mu_{0,  \text{DR}}) = 0$ we notice that 
\[
\begin{aligned}
    P(R = 0)\cdot\text{a-bias}  (\widehat \mu_{0, \text{DR}}) & = \textstyle \Ex \big [\big \{\frac{R P(R = 0) }{P(R = 1)} \omega(X, Y; \theta^\star) - {1 + R} \big \} \big \{ \tau(X,  Y) -  m_{0, \tau}(X; \theta^\star, \xi^\star)\big \} \big]\\
    & = \textstyle \Ex \big [\big \{R \big(\frac{ P(R = 0) }{P(R = 1)} \omega(X, Y; \theta^\star) + 1\big)- {1 } \big \} \big \{ \tau(X,  Y) -  m_{0, \tau}(X; \theta^\star, \xi^\star)\big \} \big]\\
    & = \textstyle \Ex \big [\big \{R \big(\frac{ P(R = 0) }{P(R = 1)} \omega(X, Y; \theta^\star) + 1\big)- {1 } \big \} \big \{ \tau(X,  Y) -  m_{0, \tau}(X; \theta^\star, \xi^\star)\big \} \big]\\
    & = \textstyle \Ex \big [\big \{\frac R{\eta_r(X, Y; \theta^\star)} - {1 } \big \} \big \{ \tau(X,  Y) -  m_{0, \tau}(X; \theta^\star, \xi^\star)\big \} \big] = \text{bias.1}
\end{aligned}
\] where the final equality holds from Lemma \ref{lemma:propensity-score}. We have already established that $\text{bias.1} = 0$, so we have $\text{a-bias}  (\widehat \mu_{0, \text{DR}}) = 0$.  
    This completes the proof.
\end{proof}

%% file: main.bbl
\begin{thebibliography}{54}
\providecommand{\natexlab}[1]{#1}
\providecommand{\url}[1]{\texttt{#1}}
\expandafter\ifx\csname urlstyle\endcsname\relax
  \providecommand{\doi}[1]{doi: #1}\else
  \providecommand{\doi}{doi: \begingroup \urlstyle{rm}\Url}\fi

\bibitem[Bang and Robins(2005)]{bang2005doubly}
H.~Bang and J.~M. Robins.
\newblock Doubly robust estimation in missing data and causal inference models.
\newblock \emph{Biometrics}, 61\penalty0 (4):\penalty0 962--973, 2005.

\bibitem[Bickel et~al.(2009)Bickel, Br{\"u}ckner, and
  Scheffer]{bickel2009discriminative}
S.~Bickel, M.~Br{\"u}ckner, and T.~Scheffer.
\newblock Discriminative learning under covariate shift.
\newblock \emph{Journal of Machine Learning Research}, 10\penalty0 (9), 2009.

\bibitem[Cai and Wei(2021)]{cai2021transfer}
T.~T. Cai and H.~Wei.
\newblock Transfer learning for nonparametric classification: Minimax rate and
  adaptive classifier.
\newblock \emph{The Annals of Statistics}, 49\penalty0 (1):\penalty0 100--128,
  2021.

\bibitem[Cheng(1994)]{cheng1994nonparametric}
P.~E. Cheng.
\newblock Nonparametric estimation of mean functionals with data missing at
  random.
\newblock \emph{Journal of the American statistical association}, 89\penalty0
  (425):\penalty0 81--87, 1994.

\bibitem[Das et~al.(2003)Das, Newey, and Vella]{das2003nonparametric}
M.~Das, W.~K. Newey, and F.~Vella.
\newblock Nonparametric estimation of sample selection models.
\newblock \emph{The Review of Economic Studies}, 70\penalty0 (1):\penalty0
  33--58, 2003.

\bibitem[Dempster et~al.(1977)Dempster, Laird, and Rubin]{dempster1977maximum}
A.~P. Dempster, N.~M. Laird, and D.~B. Rubin.
\newblock Maximum likelihood from incomplete data via the em algorithm.
\newblock \emph{Journal of the royal statistical society: series B
  (methodological)}, 39\penalty0 (1):\penalty0 1--22, 1977.

\bibitem[Deng et~al.(2009)Deng, Dong, Socher, Li, Li, and
  Fei-Fei]{deng2009imagenet}
J.~Deng, W.~Dong, R.~Socher, L.-J. Li, K.~Li, and L.~Fei-Fei.
\newblock Imagenet: A large-scale hierarchical image database.
\newblock In \emph{2009 IEEE conference on computer vision and pattern
  recognition}, pages 248--255. Ieee, 2009.

\bibitem[Efromovich(2011)]{efromovich2011nonparametric}
S.~Efromovich.
\newblock Nonparametric regression with predictors missing at random.
\newblock \emph{Journal of the American Statistical Association}, 106\penalty0
  (493):\penalty0 306--319, 2011.

\bibitem[Fay(1986)]{fay1986causal}
R.~E. Fay.
\newblock Causal models for patterns of nonresponse.
\newblock \emph{Journal of the American Statistical Association}, 81\penalty0
  (394):\penalty0 354--365, 1986.

\bibitem[Garg et~al.(2020)Garg, Wu, Balakrishnan, and Lipton]{garg2020unified}
S.~Garg, Y.~Wu, S.~Balakrishnan, and Z.~Lipton.
\newblock A unified view of label shift estimation.
\newblock \emph{Advances in Neural Information Processing Systems},
  33:\penalty0 3290--3300, 2020.

\bibitem[He et~al.(2016)He, Zhang, Ren, and Sun]{he2016deep}
K.~He, X.~Zhang, S.~Ren, and J.~Sun.
\newblock Deep residual learning for image recognition.
\newblock In \emph{Proceedings of the IEEE conference on computer vision and
  pattern recognition}, pages 770--778, 2016.

\bibitem[Heckman(1979)]{heckman1979sample}
J.~J. Heckman.
\newblock Sample selection bias as a specification error.
\newblock \emph{Econometrica: Journal of the econometric society}, pages
  153--161, 1979.

\bibitem[Horvitz and Thompson(1952)]{horvitz1952generalization}
D.~G. Horvitz and D.~J. Thompson.
\newblock A generalization of sampling without replacement from a finite
  universe.
\newblock \emph{Journal of the American statistical Association}, 47\penalty0
  (260):\penalty0 663--685, 1952.

\bibitem[Hu et~al.(2024)Hu, Yu, and Li]{hu2024receiver}
D.~Hu, T.~Yu, and P.~Li.
\newblock Receiver operating characteristic curve analysis with non-ignorable
  missing disease status.
\newblock \emph{arXiv preprint arXiv:2411.17402}, 2024.

\bibitem[Jin et~al.(2022)Jin, Ma, and Jiang]{jin2022matrix}
H.~Jin, Y.~Ma, and F.~Jiang.
\newblock Matrix completion with covariate information and informative
  missingness.
\newblock \emph{Journal of Machine Learning Research}, 23\penalty0
  (180):\penalty0 1--62, 2022.

\bibitem[Kenward and Molenberghs(1998)]{kenward1998likelihood}
M.~G. Kenward and G.~Molenberghs.
\newblock Likelihood based frequentist inference when data are missing at
  random.
\newblock \emph{Statistical Science}, pages 236--247, 1998.

\bibitem[Kim and Yu(2011)]{kim2011semiparametric}
J.~K. Kim and C.~L. Yu.
\newblock A semiparametric estimation of mean functionals with nonignorable
  missing data.
\newblock \emph{Journal of the American Statistical Association}, 106\penalty0
  (493):\penalty0 157--165, 2011.

\bibitem[Kivinen and Warmuth(1997)]{kivinen1997exponentiated}
J.~Kivinen and M.~K. Warmuth.
\newblock Exponentiated gradient versus gradient descent for linear predictors.
\newblock \emph{information and computation}, 132\penalty0 (1):\penalty0 1--63,
  1997.

\bibitem[Laan and Robins(2003)]{laan2003unified}
M.~J. Laan and J.~M. Robins.
\newblock \emph{Unified methods for censored longitudinal data and causality}.
\newblock Springer, 2003.

\bibitem[Li et~al.(2023)Li, Qin, and Liu]{li2023instability}
P.~Li, J.~Qin, and Y.~Liu.
\newblock Instability of inverse probability weighting methods and a remedy for
  nonignorable missing data.
\newblock \emph{Biometrics}, 79\penalty0 (4):\penalty0 3215--3226, 2023.

\bibitem[Lipton et~al.(2018)Lipton, Wang, and Smola]{lipton2018detecting}
Z.~Lipton, Y.-X. Wang, and A.~Smola.
\newblock Detecting and correcting for label shift with black box predictors.
\newblock In \emph{International conference on machine learning}, pages
  3122--3130. PMLR, 2018.

\bibitem[Little(1992)]{little1992regression}
R.~J. Little.
\newblock Regression with missing x's: a review.
\newblock \emph{Journal of the American statistical association}, 87\penalty0
  (420):\penalty0 1227--1237, 1992.

\bibitem[Little(1993)]{little1993pattern}
R.~J. Little.
\newblock Pattern-mixture models for multivariate incomplete data.
\newblock \emph{Journal of the American Statistical Association}, 88\penalty0
  (421):\penalty0 125--134, 1993.

\bibitem[Little(1994)]{little1994class}
R.~J. Little.
\newblock A class of pattern-mixture models for normal incomplete data.
\newblock \emph{Biometrika}, 81\penalty0 (3):\penalty0 471--483, 1994.

\bibitem[Little and Rubin(2019)]{little2019statistical}
R.~J. Little and D.~B. Rubin.
\newblock \emph{Statistical analysis with missing data}, volume 793.
\newblock John Wiley \& Sons, 2019.

\bibitem[Liu et~al.(2020)Liu, Miao, Sun, Robins, and
  Tchetgen]{liu2020identification}
L.~Liu, W.~Miao, B.~Sun, J.~Robins, and E.~T. Tchetgen.
\newblock Identification and inference for marginal average treatment effect on
  the treated with an instrumental variable.
\newblock \emph{Statistica sinica}, 30\penalty0 (3):\penalty0 1517, 2020.

\bibitem[Liu et~al.(2022)Liu, Li, and Qin]{liu2022full}
Y.~Liu, P.~Li, and J.~Qin.
\newblock Full-semiparametric-likelihood-based inference for non-ignorable
  missing data.
\newblock \emph{Statistica Sinica}, 32\penalty0 (1):\penalty0 271--292, 2022.

\bibitem[Ma et~al.(2003)Ma, Geng, and Hu]{ma2003identification}
W.-Q. Ma, Z.~Geng, and Y.-H. Hu.
\newblock Identification of graphical models for nonignorable nonresponse of
  binary outcomes in longitudinal studies.
\newblock \emph{Journal of multivariate analysis}, 87\penalty0 (1):\penalty0
  24--45, 2003.

\bibitem[Maity et~al.(2022)Maity, Sun, and Banerjee]{maity2022minimax}
S.~Maity, Y.~Sun, and M.~Banerjee.
\newblock Minimax optimal approaches to the label shift problem in
  non-parametric settings.
\newblock \emph{Journal of Machine Learning Research}, 23\penalty0
  (346):\penalty0 1--45, 2022.

\bibitem[Maity et~al.(2023)Maity, Yurochkin, Banerjee, and
  Sun]{maity2023understanding}
S.~Maity, M.~Yurochkin, M.~Banerjee, and Y.~Sun.
\newblock Understanding new tasks through the lens of training data via
  exponential tilting.
\newblock In \emph{The Eleventh International Conference on Learning
  Representations}, 2023.
\newblock URL \url{https://openreview.net/forum?id=DBMttEEoLbw}.

\bibitem[Maity et~al.(2024)Maity, Dutta, Terhorst, Sun, and
  Banerjee]{maity2024linear}
S.~Maity, D.~Dutta, J.~Terhorst, Y.~Sun, and M.~Banerjee.
\newblock A linear adjustment-based approach to posterior drift in transfer
  learning.
\newblock \emph{Biometrika}, 111\penalty0 (1):\penalty0 31--50, 2024.

\bibitem[Mao et~al.(2019)Mao, Chen, and Wong]{mao2019matrix}
X.~Mao, S.~X. Chen, and R.~K. Wong.
\newblock Matrix completion with covariate information.
\newblock \emph{Journal of the American Statistical Association}, 114\penalty0
  (525):\penalty0 198--210, 2019.

\bibitem[Miao and Tchetgen~Tchetgen(2016)]{miao2016varieties}
W.~Miao and E.~J. Tchetgen~Tchetgen.
\newblock On varieties of doubly robust estimators under missingness not at
  random with a shadow variable.
\newblock \emph{Biometrika}, 103\penalty0 (2):\penalty0 475--482, 2016.

\bibitem[Miao et~al.(2016)Miao, Ding, and Geng]{miao2016identifiability}
W.~Miao, P.~Ding, and Z.~Geng.
\newblock Identifiability of normal and normal mixture models with nonignorable
  missing data.
\newblock \emph{Journal of the American Statistical Association}, 111\penalty0
  (516):\penalty0 1673--1683, 2016.

\bibitem[Miao et~al.(2024)Miao, Liu, Li, Tchetgen~Tchetgen, and
  Geng]{miao2024identification}
W.~Miao, L.~Liu, Y.~Li, E.~J. Tchetgen~Tchetgen, and Z.~Geng.
\newblock Identification and semiparametric efficiency theory of nonignorable
  missing data with a shadow variable.
\newblock \emph{ACM/JMS Journal of Data Science}, 1\penalty0 (2):\penalty0
  1--23, 2024.

\bibitem[Nair et~al.(2019)Nair, Satpathy, Christopher,
  et~al.]{nair2019covariate}
N.~G. Nair, P.~Satpathy, J.~Christopher, et~al.
\newblock Covariate shift: A review and analysis on classifiers.
\newblock In \emph{2019 Global Conference for Advancement in Technology
  (GCAT)}, pages 1--6. IEEE, 2019.

\bibitem[Robins et~al.(2000)Robins, Rotnitzky, and
  Scharfstein]{robins2000sensitivity}
J.~M. Robins, A.~Rotnitzky, and D.~O. Scharfstein.
\newblock Sensitivity analysis for selection bias and unmeasured confounding in
  missing data and causal inference models.
\newblock In \emph{Statistical models in epidemiology, the environment, and
  clinical trials}, pages 1--94. Springer, 2000.

\bibitem[Rotnitzky et~al.(1998)Rotnitzky, Robins, and
  Scharfstein]{rotnitzky1998semiparametric}
A.~Rotnitzky, J.~M. Robins, and D.~O. Scharfstein.
\newblock Semiparametric regression for repeated outcomes with nonignorable
  nonresponse.
\newblock \emph{Journal of the american statistical association}, 93\penalty0
  (444):\penalty0 1321--1339, 1998.

\bibitem[Rubin(1976)]{rubin1976inference}
D.~B. Rubin.
\newblock Inference and missing data.
\newblock \emph{Biometrika}, 63\penalty0 (3):\penalty0 581--592, 1976.

\bibitem[Rubin(2004)]{rubin2004multiple}
D.~B. Rubin.
\newblock \emph{Multiple imputation for nonresponse in surveys}, volume~81.
\newblock John Wiley \& Sons, 2004.

\bibitem[Sagawa et~al.(2020)Sagawa, Koh, Hashimoto, and
  Liang]{Sagawa2020Distributionally}
S.~Sagawa, P.~W. Koh, T.~B. Hashimoto, and P.~Liang.
\newblock Distributionally robust neural networks.
\newblock In \emph{International Conference on Learning Representations}, 2020.
\newblock URL \url{https://openreview.net/forum?id=ryxGuJrFvS}.

\bibitem[Schenker and Welsh(1988)]{schenker1988asymptotic}
N.~Schenker and A.~H. Welsh.
\newblock Asymptotic results for multiple imputation.
\newblock \emph{The Annals of Statistics}, 16\penalty0 (4):\penalty0
  1550--1566, 1988.

\bibitem[Scott(2019)]{scott2019generalized}
C.~Scott.
\newblock A generalized neyman-pearson criterion for optimal domain adaptation.
\newblock In \emph{Algorithmic Learning Theory}, pages 738--761. PMLR, 2019.

\bibitem[Sugiyama et~al.(2007)Sugiyama, Krauledat, and
  M{\"u}ller]{sugiyama2007covariate}
M.~Sugiyama, M.~Krauledat, and K.-R. M{\"u}ller.
\newblock Covariate shift adaptation by importance weighted cross validation.
\newblock \emph{Journal of Machine Learning Research}, 8\penalty0 (5), 2007.

\bibitem[Sun et~al.(2018)Sun, Liu, Miao, Wirth, Robins, and
  Tchetgen]{sun2018semiparametric}
B.~Sun, L.~Liu, W.~Miao, K.~Wirth, J.~Robins, and E.~J.~T. Tchetgen.
\newblock Semiparametric estimation with data missing not at random using an
  instrumental variable.
\newblock \emph{Statistica Sinica}, 28\penalty0 (4):\penalty0 1965, 2018.

\bibitem[Tang and Ju(2018)]{tang2018statistical}
N.~Tang and Y.~Ju.
\newblock Statistical inference for nonignorable missing-data problems: a
  selective review.
\newblock \emph{Statistical Theory and Related Fields}, 2\penalty0
  (2):\penalty0 105--133, 2018.

\bibitem[Tchetgen~Tchetgen and Wirth(2017)]{tchetgen2017general}
E.~J. Tchetgen~Tchetgen and K.~E. Wirth.
\newblock A general instrumental variable framework for regression analysis
  with outcome missing not at random.
\newblock \emph{Biometrics}, 73\penalty0 (4):\penalty0 1123--1131, 2017.

\bibitem[Tsiatis(2006)]{tsiatis2006semiparametric}
A.~A. Tsiatis.
\newblock \emph{Semiparametric theory and missing data}, volume~4.
\newblock Springer, 2006.

\bibitem[Wah et~al.(2011)Wah, Branson, Welinder, Perona, and
  Belongie]{wah2011caltech}
C.~Wah, S.~Branson, P.~Welinder, P.~Perona, and S.~Belongie.
\newblock The caltech-ucsd birds-200-2011 dataset.
\newblock 2011.

\bibitem[Wang et~al.(2021)Wang, Shao, and Fang]{wang2021propensity}
L.~Wang, J.~Shao, and F.~Fang.
\newblock Propensity model selection with nonignorable nonresponse and
  instrument variable.
\newblock \emph{Statistica Sinica}, 31\penalty0 (2):\penalty0 647--672, 2021.

\bibitem[Wang et~al.(2014)Wang, Shao, and Kim]{wang2014instrumental}
S.~Wang, J.~Shao, and J.~K. Kim.
\newblock An instrumental variable approach for identification and estimation
  with nonignorable nonresponse.
\newblock \emph{Statistica Sinica}, pages 1097--1116, 2014.

\bibitem[Yu et~al.(2018)Yu, Kim, and Park]{yu2018estimation}
W.~Yu, J.~K. Kim, and T.~Park.
\newblock Estimation of area under the roc curve under nonignorable
  verification bias.
\newblock \emph{Statistica Sinica}, 28\penalty0 (4):\penalty0 2149, 2018.

\bibitem[Zhou et~al.(2017)Zhou, Lapedriza, Khosla, Oliva, and
  Torralba]{zhou2017places}
B.~Zhou, A.~Lapedriza, A.~Khosla, A.~Oliva, and A.~Torralba.
\newblock Places: A 10 million image database for scene recognition.
\newblock \emph{IEEE transactions on pattern analysis and machine
  intelligence}, 40\penalty0 (6):\penalty0 1452--1464, 2017.

\bibitem[Zhu et~al.(2022)Zhu, Wang, and Samworth]{zhu2022high}
Z.~Zhu, T.~Wang, and R.~J. Samworth.
\newblock High-dimensional principal component analysis with heterogeneous
  missingness.
\newblock \emph{Journal of the Royal Statistical Society Series B: Statistical
  Methodology}, 84\penalty0 (5):\penalty0 2000--2031, 2022.

\end{thebibliography}
